\documentclass{article}

\newcommand{\sv}[1]{#1}
\newcommand{\lv}[1]{}

\newcommand{\icml}[1]{}
\newcommand{\arxiv}[1]{#1}

\newcommand{\comment}[1]{}

\arxiv{\usepackage{natbib}
\setcitestyle{authoryear,round,citesep={;},aysep={,},yysep={;}}
 
\renewcommand{\cite}[1]{\citep{#1}}
\usepackage{url}
\usepackage{fullpage}
\usepackage{authblk}
}
\usepackage{microtype}
\usepackage[scaled=0.90]{helvet}

\usepackage{graphicx}
\usepackage{booktabs}
\usepackage{caption}
\usepackage{subcaption}
\usepackage{array}
\usepackage{amsthm,amsmath,amssymb,bbm,graphicx,rotating,verbatim}
\usepackage{boxedminipage}
\usepackage{enumerate}
\usepackage{epic,eepic}
\usepackage{tikz}
\usetikzlibrary{shapes}
\usepackage[rightcaption]{sidecap}
\usepackage{booktabs}
\usepackage{graphicx}
\usepackage{paralist}
\usepackage{floatrow}
\usepackage[author=,draft]{fixme}
\fxsetup{theme=color}

\newcommand{\nb}[1]{{\color{red} \scriptsize [ #1 ]}}
\newcommand{\ie}{\emph{i.e.}}
\newcommand{\eg}{\emph{e.g.}}
\newcommand{\etal}{\emph{et al.}}
\newcommand{\FPT}{\textbf{\textup{FPT}}}
\newcommand{\XP}{\textbf{\textup{XP}}}
\newcommand{\RFPT}{\textbf{\textup{FPT$_R$}}}
\newcommand{\RXP}{\textbf{\textup{XP$_R$}}}

\newcommand{\Weft}{{\textbf{\textup{W}}}}
\newcommand{\W}[1]{{\Weft}{{[#1]}}}
\newcommand{\paraNP}{\textbf{\textup{paraNP}}}
\newcommand{\NP}{\textbf{\textup{NP}}}
\newcommand{\hy}{\hbox{-}\nobreak\hskip0pt}

\newcommand{\coloring}{{\sc ($n/3$)-Coloring$^{\delta = n-4}$}}

\newcommand{\bigO}[1]{\ensuremath{{\mathcal O}(#1)}}

\newcommand{\SB}{\{\,}%
\newcommand{\SM}{\;{|}\;}%
\newcommand{\SE}{\,\}}%

\newcommand{\uDRM}{\textsc{DRMC}}
\newcommand{\bDRM}{\textsc{$p$-DRMC}}
\newcommand{\uRM}{\textsc{RMC}}
\newcommand{\bRM}{\textsc{$p$-RMC}}

\newcommand{\matM}{\mathbf{M}}
\newcommand{\rk}{\textsf{rk}}
\newcommand{\ur}{\textsf{dr}}
\newcommand{\GF}{\text{GF}}
\newcommand{\EQ}{\text{EQ}}
\newcommand{\col}{\textsf{col}}
\newcommand{\row}{\textsf{row}}
\newcommand{\comb}{\textsf{comb}}

\newcommand{\MEnt}{\mathbf{E}}

\newcommand{\MSO}{MSO$_2$}
\newcommand{\textfor}[1]{\mathit{#1}}
\newcommand{\tw}{\mathsf{tw}}

\newcommand{\PPP}{\ensuremath{{\mathcal P}}}
\newcommand{\III}{\ensuremath{{\mathcal I}}}
\newcommand{\RRR}{\ensuremath{{\mathcal R}}}

\theoremstyle{plain}
\newtheorem{THE}{Theorem}
\newtheorem{OBS}{Observation}
\newtheorem{PRO}[THE]{Proposition}
\newtheorem{LEM}[THE]{Lemma}
\newtheorem{COR}[THE]{Corollary}
\newtheorem{CLM}{Claim}

\newcommand{\pbDef}[3]{%
\noindent
\begin{center}
\begin{boxedminipage}{1 \columnwidth}
#1\\[5pt]
\begin{tabular}{l p{0.73 \columnwidth}}
Input: & #2\\
Task: & #3
\end{tabular}
\end{boxedminipage}
\end{center}
}

\icml{\usepackage{hyperref}}

\icml{\usepackage{icml2018}}
\begin{document}

\arxiv{
\title{{Parameterized Algorithms for the Matrix Completion Problem}}
\author[1]{Robert Ganian}
\author[2]{Iyad Kanj}
\author[1]{Sebastian Ordyniak}
\author[1]{Stefan Szeider}

\affil[1]{Algorithms and Complexity Group, TU Wien, Austria}
\affil[2]{School of Computing, DePaul University, Chicago, USA}
\renewcommand\Authands{ and }

\date{}

\maketitle
}

\icml{
\twocolumn[
\icmltitle{Parameterized Algorithms for the Matrix Completion Problem}

\icmlsetsymbol{equal}{*}

\begin{icmlauthorlist}
\icmlauthor{Robert Ganian}{tuw}
\icmlauthor{Iyad Kanj}{chic}
\icmlauthor{Sebastian Ordyniak}{tuw}
\icmlauthor{Stefan Szeider}{tuw}
\end{icmlauthorlist}

\icmlaffiliation{tuw}{Algorithms and Complexity Group, TU Wien, Austria}
\icmlaffiliation{chic}{School of Computing, DePaul University, Chicago, USA}

\icmlcorrespondingauthor{Robert Ganian}{rganian@gmail.com}
\icmlcorrespondingauthor{Iyad Kanj}{ikanj@cdm.depaul.edu}
\icmlcorrespondingauthor{Sebastian Ordyniak}{sordyniak@gmail.com}
\icmlcorrespondingauthor{Stefan Szeider}{stefan@szeider.net}

\icmlkeywords{matrix completion, parameterized complexity, exact algorithms, lower bounds}

\vskip 0.3in
]

\printAffiliationsAndNotice{\icmlEqualContribution}
}

%\begin{itemize}
%\item Submissions must be in PDF\@.
%\item The maximum paper length is \textbf{8 pages excluding references and
%    acknowledgements, and 10 pages including references and acknowledgements}
%    (pages 9 and 10 must contain only references and acknowledgements).
%\item \textbf{Do not include author information or acknowledgements} in your
%    initial submission.
%\item Your paper should be in \textbf{10 point Times font}.
%\item Make sure your PDF file only uses Type-1 fonts.
%\item Place figure captions \emph{under} the figure (and omit titles from inside
%    the graphic file itself). Place table captions \emph{over} the table.
%\item References must include page numbers whenever possible and be as complete
%    as possible. Place multiple citations in chronological order.
%\item Do not alter the style template; in particular, do not compress the paper
%    format by reducing the vertical spaces.
%\item Keep your abstract brief and self-contained, one paragraph and roughly
%    4--6 sentences. Gross violations will require correction at the
%    camera-ready phase. The title should have content words capitalized.
%\end{itemize}
% NO APPENDIX
%Abstracts must be a single paragraph, ideally between 4--6 sentences long.

\begin{abstract}
We consider two matrix completion problems, in which we are given a matrix with missing entries and the task is to complete the matrix in a way that (1) minimizes the rank, or (2) minimizes the number of distinct rows.
%These problems have ubiquitous applications in many areas, including recommender systems, machine learning, sensing, computer vision, data science, and predictive analytics, among others. In particular, the first problem is a fundamental problem that has undergone a tremendous amount of research.
We study the parameterized complexity of the two aforementioned problems with respect to several parameters of interest, including the minimum number of matrix rows, columns, and rows plus columns needed to cover all missing entries. We obtain new algorithmic results showing that, for the bounded domain case, both problems are fixed-parameter tractable with respect to all aforementioned parameters. We complement these results with a lower-bound result for the unbounded domain case that rules out fixed-parameter tractability w.r.t.\ some of the parameters under consideration.
\end{abstract}

\section{Introduction}
\label{sec:intro}

\textbf{Problem Definition and Motivation}. We consider the matrix completion problem, in which we are given a matrix $\matM$ (over some field that we also refer to as the {\em domain} of the matrix) with missing entries, and the goal is to complete the entries of $\matM$ so that to optimize a certain measure. There is a wealth of research on this fundamental problem~\cite{cp10,cr09,ct10,ev13,hmrw14,fazel01,kmo10,kmo101,recht11,sfh16} due to its ubiquitous applications in recommender systems, machine learning, sensing, computer vision, data science, and predictive analytics, among others. In these areas, the matrix completion problem naturally arises after observing a sample from the set of entries of a low-rank matrix, and attempting to recover the missing entries with the goal of optimizing a certain measure. In this paper, we focus our study on matrix completion with respect to two measures (considered separately): (1) minimizing the rank of the completed matrix, and (2) minimizing the number of distinct rows of the completed matrix.

The first problem we consider---matrix completion w.r.t.~rank minimization---has been extensively studied, and is often referred to as the low-rank matrix completion problem~\cite{cp10,cr09,ct10,hmrw14,fazel01,kmo10,kmo101,recht11,sfh16}.
A celebrated application of this problem lies in the recommender systems area, where it is known as the Netflix problem~\cite{netflix}. In this user-profiling application, an entry of the input matrix represents the rating of a movie by a user, where some entries could be missing.
%, and otherwise, the entry is missing.
The goal is to predict the missing entries so that the rank of the complete matrix is minimized.

The low-rank matrix completion problem is known to be \NP-hard, even when the matrix is over the field~$\GF(2)$ (\ie, each entry is 0 or 1), and the goal is to complete the matrix into one of rank 3~\cite{peeters96}.
A significant body of work on the low-rank matrix completion problem has centered around proving that, under some feasibility assumptions, the matrix completion problem can be solved efficiently with high probability~\cite{cr09,recht11}. These feasibility assumptions are: (1) low rank; (2) incoherence; and (3) randomness~\cite{hmrw14}.
Hardt \etal~(\citeyear{hmrw14}) argue that feasibility assumption (3), which
states that the subset of determined entries in the matrix is selected
uniformly at random and has a large (sampling) density, is very
demanding. In particular, they justify that in many applications, such as the
Netflix problem, it is not possible to arbitrarily choose
which matrix entries are determined and which are not, as those may be dictated by outside factors.
The low-rank matrix completion problem also has other applications in the area of wireless sensor networks. In one such application, the goal is to reconstruct a low-dimensional geometry describing the locations of the sensors based on local distances sensed by each sensor; this problem is referred to as \textsc{triangulation from incomplete data}~\cite{cr09}. Due to its inherent hardness, the low-rank matrix completion problem has also been studied with respect to various notions of approximation~\cite{cr09,ct10,fkv04,hmrw14,kmo10,kmo101,recht11}.
%One particular such notion which has received considerable attention is: Given an incomplete matrix $\matM$ and an error bound $\epsilon > 0$, compute a complete matrix $\matM'$ of minimum rank such that $||\matM - \matM'|| \leq \epsilon$, where $||\matM - \matM'||$ stands for the mean squared error between the known entries in $\matM$ and their corresponding entries in $\matM'$~\cite{hmrw14,kmo10}.

The second problem we consider is the matrix completion problem w.r.t.~minimizing the number of distinct rows. Although this problem has not received as much attention as low-rank-matrix completion, it certainly warrants studying. In fact, minimizing the number of distinct rows represents a special case of the \textsc{sparse subspace clustering} problem~\cite{ev13}, where the goal is to complete a matrix in such a way that its rows can be partitioned into the minimum number of subspaces. The problem we consider corresponds to the special case of \textsc{sparse subspace clustering} where the matrix is over $\GF(2)$ and the desired rank of each subspace is 1.
%We note that, as alluded to above, the (\NP-hard) matrix completion problem over $\GF(2)$ is a well-studied problem (\eg, see~\cite{peeters96}).
%, as it is capable of encoding some graph coloring problems.
%In addition to its relation to the sparse subspace clustering problem,
Furthermore, one can see the relevance of this problem to the area of recommender systems; in this context, one seeks to complete the matrix in such a way that the profile of each user is identical to a member of a known (possibly small) group of users.

In this paper, we study the two aforementioned problems through the lens of \emph{parameterized complexity}~\cite{DowneyFellows13}. In this paradigm, one measures the complexity of problems not only in terms of their input size $n$ but also by a certain \emph{parameter} $k\in \mathbb{N}$, and seeks---among other things---fixed-parameter algorithms, i.e., algorithms that run in time $f(k)\cdot n^{\bigO{1}}$ for some function $f$.
Problems admitting such algorithms are said to be
\emph{fixed-parameter tractable} (or contained in the parameterized
complexity class \FPT{}). The motivation is that the parameter of choice---usually describing some structural properties of the instance---can be small in some instances of interest, even when the input size is large. Therefore, by confining the combinatorial explosion to this parameter, one can obtain efficient algorithms for problem instances with a small parameter value for \NP-hard problems.
% For problems that are not fixed-parameter tractable there is still
% hope that they can be solved in in polynomial-time for each
% fixed parameter value.
Problems that are not (or unlikely to be) fixed-parameter tractable
can still be solvable in polynomial-time for every fixed parameter
value, i.e., they can be solved in time $n^{f(k)}$ for some function
$f$. Problems of this kind are contained in the
parameterized complexity class \XP{}.
We also consider randomized versions of \FPT{} and \XP{}, denoted by \RFPT{} and \RXP{}, containing
all problems that can be solved by a randomized algorithm with a
run-time of $f(k)n^{\bigO{1}}$ and $\bigO{n^{f(k)}}$, respectively, with a
constant one-sided error-probability.
Finally, problems that remain
\NP\hy hard for some fixed value of the parameter are hard for the
parameterized complexity class \paraNP{}.
We refer to the respective textbooks for a detailed introduction to parameterized complexity~\cite{DowneyFellows13,CyganFKLMPPS15}.
% Apart from \FPT, a problem with a chosen parameterization can also belong to the class \XP\ (meaning that it can be solved by an \XP\ algorithm, i.e., in time $n^{f(k)}$ for some function $f$) or be \paraNP-hard (meaning that the problem remains \NP-hard even for a fixed value of the parameter).
Parameterized Complexity is a rapidly growing field with various
applications in many areas of Computer Science, including Artificial
Intelligence~\cite{GottlobPichlerWei10,BevernKNSW16,GanianOrdyniak18,BessiereHHKQW08,BonnetGLRS17}.

\textbf{Parameterizations}. The parameters that we consider in this
paper are: The number of (matrix) rows that cover all missing entries
(\row{}); the number of columns that cover all missing entries
(\col{}); and the minimum number of rows and columns which together
cover all missing entries (\comb{}). Although we do discuss and
provide results for the unbounded domain case, i.e, the case that the
domain (field size) is part of the input, we focus on the case when
the matrix is over a bounded domain: This case is the most relevant
from a practical perspective, and most of the related works focus on
this case. It is easy to see that, when stated over any bounded domain, both
problems under consideration are in \FPT{} when parameterized by the
number of missing entries, since an algorithm can brute-force through
all possible solutions.
%
% todo: resolve stupid formatting issue
On the other hand, parameterizing by \row{}
(resp.~\col{}) is very interesting from a practical perspective, as
rows (resp.~columns) with missing entries represent the newly-added
elements (\eg, newly-added users/movies/sensors, etc.); here, the
above brute-force approach naturally fails, since the number of
missing entries is no longer bounded by the parameter alone.
Finally, the parameterization by \comb{} is interesting because this
parameter subsumes (\ie, is smaller than) the other two parameters
(\ie, \row{} and \col{}). In particular, any fixed-parameter algorithm
w.r.t.\ this parameter implies a fixed-parameter algorithm w.r.t.\ the
other two parameters, but can also remain efficient in cases where the
number of rows and/or columns with missing entries is large.
% We discuss the results of our study in the next section.

\textbf{Results and Techniques}. We start in Section~\ref{sec:rm} by
considering the {\sc Bounded Rank Matrix Completion} problem over
$\GF(p)$ (denoted \textsc{$p$-RMC}), in which the goal is to complete
the missing entries in the input matrix so that the rank of the
completed matrix is at most $t$, where $t \in \mathbb{N}$ is given as
input. We present a (randomized) fixed-parameter algorithm for this problem
parameterized by \comb{}. This result is obtained by applying a
branch-and-bound algorithm combined with algebra techniques, allowing
us to reduce the problem to a system of quadratic equations in which only few
(bounded by some function of the parameter)
equations contain non-linear terms. We then use a result by
Miura \etal~(\citeyear{MiuraHT14}) (improving an earlier result by
Courtois \etal~(\citeyear{CourtoisGMT02}))
in combination with reduction techniques to
show that solving such a system of
equations is in \RFPT{} parameterized by the number of equations
containing non-linear terms.
In the case where the domain is unbounded, %(\textsc{RMC} denotes this variant of {\sc Bounded Rank Matrix Completion}),
we show that \textsc{RMC} is in \XP{} parameterized by either \row{} or \col{} and in \RXP{}
parameterized by \comb{}.
%
%   We remark that the aforementioned algorithm immediately gives \FPT{} algorithms with the same running time for {\sc $p$-RMC} parameterized by the (larger or equal parameter) number of covering rows (resp.~columns). However, for the parameterization by the number of rows (resp.~columns), $k$, we present a simpler and faster \FPT{} algorithm running in time $\bigO{2^kp^{k^2}\cdot (m+kn)^{2.4}}$.
%; the algorithm runs in time $\bigO{p^{1.1^{k^2}}} (n+m)^{\bigO{1}}$, where $n$ and $m$ are the dimensions of the input matrix. This result is obtained by applying a branch-and-bound algorithm combined with algebra techniques, allowing us to reduce the problem to a system of equations in which \FPT{}-many equations contain quadratic terms. We then appeal to a result by Courtois \etal~\cite{CourtoisGMT02} showing that such a system of equations can be solved in \FPT-time. We remark that the aforementioned algorithm immediately gives \FPT{} algorithms with the same running time for {\sc $p$-RMC} parameterized by the (larger or equal parameter) number of covering rows (resp.~columns). However, for the parameterization by the number of rows (resp.~columns), $k$, we present a simpler and faster \FPT{} algorithm running in time $\bigO{2^kp^{k^2}\cdot (m+kn)^{2.4}}$.

In Section~\ref{sec:fr}, we turn our attention to the {\sc Bounded Distinct Row Matrix Completion} problem over both bounded domain ({\sc $p$-DRMC}) and unbounded domain ({\sc DRMC}); here, the goal is to complete the input matrix so that the number of distinct rows in the completed matrix is at most $t$. We start by showing that {\sc $p$-DRMC} parameterized by \comb{} is fixed-parameter tractable.
We obtain this result as a special case of a more general result showing that both {\sc DRMC} and {\sc $p$-DRMC} are fixed-parameter tractable parameterized by the \emph{treewidth}~\cite{RobertsonSeymour86,DowneyFellows13} of the compatibility graph, i.e., the graph having one vertex for every row and an edge between two vertices if the associated rows can be made identical. This result also allows us to show that {\sc DRMC} is fixed-parameter tractable parameterized by \row{}. Surprisingly, {\sc DRMC} behaves very differently when parameterized by \col{}, as we show that, for this parameterization, the problem becomes \paraNP{}\hy hard.
% To obtain this result, we introduce the \emph{compatibility graph} of an instance of the problem, and prove a generic result stating that both {\sc DRMC} and {\sc $p$-DRMC} are \FPT{} when parameterized by the \emph{treewidth}~\cite{RobertsonSeymour86,DowneyFellows13} of the compatibility graph. We then employ this generic result to obtain a fixed-parameter algorithm for {\sc $p$-DRMC} parameterized by the minimum number of rows plus columns that cover all missing entries, and to obtain a fixed-parameter algorithm for {\sc DRMC} parameterized by the minimum number of rows that cover all missing entries.
% %; we note that the former result immediately implies that {\sc $p$-DRMC} is \FPT{} parameterized by the minimum number of rows or columns that cover all missing entries.
% We contrast the \FPT{} result for {\sc $p$-DRMC} parameterized by the minimum number of rows that cover all missing entries, with a result showing that {\sc $p$-DRMC}, parameterized by the minimum number of columns that cover all missing entries, is para-\NP-hard (\ie, is \NP-hard even for some constant value of the parameter)\nb{Iyad: Why repeat what paraNP-hard means? We've already defined it earlier.}; we do so via an involved reduction from a restriction of the {\sc Partition Into Triangles} problem.

\begin{table}[ht]
\vspace{-0.1cm}
% \begin{tabular}{r|ccc}
%     & \row & \col & \comb \\\toprule
%     \bRM{} & \multicolumn{3}{c}{\FPT{} (Th.~\ref{thm:rbcomb})} \\
%     \bDRM{} & \multicolumn{3}{c}{\FPT{} (Th.~\ref{thm:bounded-drm-fpt})} \\\midrule
%     \uRM{} & \multicolumn{2}{c}{\hspace{0.8cm}\XP{} (Cor.~\ref{cor:XP})} & $?$ \\
%     \uDRM{} & \FPT{} (Th.~\ref{thm:unbounded-drm-fpt}) & \multicolumn{2}{c}{\paraNP{} (Th.~\ref{thm:unbounded-drm-np})} \\
%   \end{tabular}
% \bigskip \nb{below is the table in an alternative format; I hope it's
%   clearer, although more redundant -S }
\begin{tabular}{@{}l@{\quad}l@{~~~}l@{~~~}l@{}} \toprule
    & \row & \col & \comb \\\midrule
\bRM{} & \FPT{}${}^\text{(Th.~\ref{the:mcrbrow})}$ &
       \FPT{}${}^\text{(Cor.~\ref{cor:mcrbcol})}$ &
       \RFPT{}${}^\text{(Th.~\ref{thm:rbcomb})}$ \\
\bDRM{} & \FPT{}${}^\text{(Th.~\ref{thm:bounded-drm-fpt})}$ &
       \FPT{}${}^\text{(Th.~\ref{thm:bounded-drm-fpt})}$ &
       \FPT{}${}^\text{(Th.~\ref{thm:bounded-drm-fpt})}$ \\ \midrule
\uRM{} & \XP{}${}^\text{(Cor.~\ref{cor:XP})}$ &
       \XP{}${}^\text{(Cor.~\ref{cor:XP})}$ &
       \RXP{}${}^\text{(Cor.~\ref{cor:combxp})}$ \\
\uDRM{} & \FPT{}${}^\text{(Th.~\ref{thm:unbounded-drm-fpt})}$ &
       \paraNP{}${}^\text{(Th.~\ref{thm:unbounded-drm-np})}$ &
       \paraNP{}${}^\text{(Th.~\ref{thm:unbounded-drm-np})}$\\ \bottomrule
 \end{tabular}
  \vspace{-0.1cm}
  \caption[]{The parameterized complexity results obtained for the problems \bRM{} and \bDRM{} and their unbounded domain variants \uRM{} and \uDRM{} w.r.t.\ the parameters \row{}, \col{}, \comb{}. %The question mark ($?$) in the table indicates that the complexity of \uRM{} parameterized by \comb{} is open.
  \vspace{-0.2cm}}
  \label{tab:1}
\end{table}

We chart our results in Table~\ref{tab:1}. Interestingly, in the unbounded domain case, both considered problems exhibit wildly different behaviors: While \textsc{RMC} admits \XP\ algorithms regardless of whether we parameterize by \row{} or \col{}, using these two parameterizations for \textsc{DRMC} results in the problem being \FPT{} and \paraNP\hy hard, respectively. On the other hand, in the (more studied) bounded domain case, we show that both problems are in \FPT{} (resp.~\RFPT{}) w.r.t.~all parameters under consideration. Finally, we prove that {\sc $2$-DRMC} remains \NP-hard even if every column and row contains (1) a bounded number of missing entries, or (2) a bounded number of determined entries. This effectively rules out \FPT{} algorithms w.r.t.~the parameters: maximum number of missing/determined entries per row or column.

%even if every column
% and every row contains exactly 3 missing entries, or if every column and every row contains at most . %This result implies that {\sc $p$-DRMC}, parameterized by the maximum number of missing entries per row and per column, is para-\NP-hard.

%Our problem is called ``Matrix Completion'', well-studied, has a nice wikipedia page too:
%https://en.wikipedia.org/wiki/Matrix_completion
%there they explicitly say that ``matrix completion often seeks to find the lowest rank matrix or, if the rank of the completed matrix is known, a matrix of rank {\displaystyle r} r that matches the known entries.''
%We should be able to find good references for the above
%We'll have to then argue that another natural way of completing a matrix is to seek to minimize the number of distinct rows. Can probably argue with the Netflix problem again --- if say you have a bunch of new users or movies (new rows or columns), then it is natural to expect that these are likely to behave similarly as known users/movies based on their already known entries
%We also need to argue about us considering bounded/unbounded domain... hopefully both cases make sense.

\section{Preliminaries}
\label{sec:prel}
For a prime number $p$, let $\GF(p)$ be a field of order $p$; recall that each such field can be equivalently represented as the set of integers modulo $p$. For positive integers $i$ and $j>i$, we write $[i]$ for the set $\{1,2,\dots,i\}$, and $i:j$ for the set $\{i,i+1,\dots,j\}$.

For an $m\times n$ matrix $\matM$ (\ie, a matrix with $m$ rows and $n$ columns), and for $i\in [m]$ and $j\in [n]$, $\matM[i,j]$ denotes the element in the $i$-th row and $j$-th column of $\matM$. Similarly, for a vector $d$, we write $d[i]$ for the $i$-th coordinate of $d$.
We write $\matM[*,j]$ for the \emph{column-vector} $(\matM[1,j],\matM[2,j],\dots,\matM[m,j])$, and $\matM[i,*]$ for the \emph{row-vector} $(\matM[i,1],\matM[i,2],\dots,\matM[i,n])$. We will also need to refer to submatrices obtained by omitting certain rows or columns from $\matM$. We do so by using sets of indices to specify which rows and columns the matrix contains. For instance, the matrix $\matM[[i],*]$ is the matrix consisting of the first $i$ rows and all columns of $\matM$, and $\matM[2:m,1:n-1]$ is the matrix obtained by omitting the first row and the last column from $\matM$.

The \emph{row-rank} (resp.~\emph{column-rank}) of a matrix $\matM$ is the maximum number of linearly-independent rows (resp.~\emph{columns}) in $\matM$. It is well known that the row-rank of a matrix is equal to its column-rank, and this number is referred to as the \emph{rank} of the matrix.  We let $\rk(\matM)$ and $\ur(\matM)$ denote the rank and the number of distinct rows of a matrix $\textbf{M}$, respectively. If $\matM$ is a matrix over $\GF(p)$, we call $\GF(p)$ the \emph{domain} of $\matM$.

An \emph{incomplete matrix} over $\GF(p)$ is a matrix which may contain not only elements from $\GF(p)$ but also the special symbol $\bullet$. An entry is a \emph{missing} entry if it contains $\bullet$, and is a \emph{determined} entry otherwise. A (possibly incomplete) $m\times n$ matrix $\matM'$ is \emph{consistent} with an $m\times n$ matrix $\matM$ if and only if, for each $i\in [m]$ and $j\in [n]$, either $\matM'[i,j]=\matM[i,j]$ or $\matM'[i,j]=\bullet$.

\subsection{Problem Formulation}
\label{sub:problems}

%\nb{We should state the assumption that we can always assume that all  rows in \uDRM{} are distinct}

We formally define the problems under consideration below.

\pbDef{\small\textsc{Bounded Rank Matrix Completion} ($p$-\textsc{RMC})}{An incomplete matrix $\matM$ over $\GF(p)$ for a fixed prime number $p$, and an integer $t$.}{Find a matrix $\matM'$ consistent with $\matM$ such that $\rk(\matM')\leq t$.}

\pbDef{\small\textsc{Bounded Distinct Row Matrix Completion} ($p$-\textsc{DRMC})}{An incomplete matrix $\matM$ over $\GF(p)$ for a fixed prime number $p$, and an integer $t$.}{Find a matrix $\matM'$ consistent with $\matM$ such that $\ur(\matM')\leq t$.}

%It is worth noting that for \textsc{DRMC}, the field $\GF(p)$ can be equivalently viewed as an arbitrary set of arity $p$.
Aside from the problem variants where $p$ is a fixed prime number, we also study the case where matrix entries range over a domain that is provided as part of the input. In particular, the problems \textsc{RMC} and \textsc{DRMC} are defined analogously to \textsc{$p$-RMC} and \textsc{$p$-DRMC}, respectively, with the sole distinction that the prime number $p$ is provided as part of the input.
%given as part of the input, we will also study the case where matrix entries range only over a fixed (\ie, constant-size) domain. The problems \textsc{$p$-RMC} and \textsc{$p$-DRMC} are defined analogously to \textsc{RMC} and \textsc{DRMC}, respectively, with the sole distinction that $p$ is considered to be a fixed (constant) prime number that is not part of the input.
We note that \textsc{$2$-RMC} is \NP-hard even for $t=3$~\cite{peeters96}, and the same holds for \textsc{$2$-DRMC} (see Theorem~\ref{thm:drmchard}). Without loss of generality, we assume that the rows of the input matrix are pairwise distinct.

\lv{
\subsection{Parameterized Complexity}

The class \FPT{} consists of all parameterized problems solvable in time $f(k)\cdot n^{\bigO{1}}$, where $k$ is the parameter and $n$ is the input size; we refer to such running time as \emph{\FPT{}-time}.
The class \XP\ contains all parameterized problems solvable in time $\bigO{n^{f(k)}}$.
The following relations hold among the parameterized complexity classes: \mbox{$\FPT \subseteq \W{1}\subseteq \XP$}.

The class \paraNP\ is defined as the class of parameterized problems that are solvable by a non-deterministic Turing machine in \FPT{}-time.
In the \paraNP-hardness proofs, we will make use of the following characterization of \paraNP-hardness~\cite{FlumGrohe06}: Any parameterized problem that remains \NP-hard when the
parameter is a constant is \paraNP-hard. For problems in \NP, it holds that $\XP \subseteq \paraNP$; in particular showing the \paraNP-hardness of a problem rules out the existence of algorithms running in time  $\bigO{n^{f(k)}}$ for the problem.
We also consider randomized versions of \FPT{} and \XP{}, denoted by \RFPT{} and \RXP{}, containing
all problems that can be solved by a randomized algorithm with a
run-time of $f(k)n^{\bigO{1}}$ and $\bigO{n^{f(k)}}$, respectively, with a
constant one-sided error-probability.
}

\subsection{Treewidth}
\label{sub:tw}
\sv{Treewidth~\cite{RobertsonSeymour86} is one of the most prominent
decompositional parameters for graphs and has found numerous
applications in computer science.
A \emph{tree-decomposition}~$\mathcal{T}$ of a graph $G=(V,E)$ is a pair
$(T,\chi)$, where $T$ is a tree and $\chi$ is a function that assigns each
tree node $t$ a set $\chi(t) \subseteq V$ of vertices such that the following
conditions hold: (TD1) for every edge $uv \in E(G)$ there is a tree
node $t$ such that $u,v\in \chi(t)$; and (TD2) for every vertex $v
\in V(G)$, the set of tree nodes $t$ with $v\in \chi(t)$ forms a
non-empty subtree of~$T$. The \emph{width} of a tree-decomposition
$(T,\chi)$ is the size of a largest bag minus~$1$.  A tree-decomposition of
minimum width is called \emph{optimal}. The \emph{treewidth} of a graph $G$,
denoted by $\tw(G)$, is the width of an optimal tree decomposition
of~$G$. We will assume that
the tree $T$ of a tree-decomposition is rooted and we will denote by $T_t$ the subtree of $T$ rooted
at $t$ and write $\chi(T_t)$ for the set $\bigcup_{t'\in
  V(T_t)}\chi(t')$.
}

\lv{Treewidth~\cite{RobertsonSeymour86} is one of the most prominent
decompositional parameters for graphs and has found numerous
applications in computer science.
A \emph{tree-decomposition}~$\mathcal{T}$ of a graph $G=(V,E)$ is a pair
$(T,\chi)$, where $T$ is a tree and $\chi$ is a function that assigns each
tree node $t$ a set $\chi(t) \subseteq V$ of vertices such that the following
conditions hold:
\begin{itemize}
\item[(TD1)] For every edge $uv \in E(G)$ there is a tree node
  $t$ such that $u,v\in \chi(t)$.
\item[(TD2)] For every vertex $v \in V(G)$,
  the set of tree nodes $t$ with $v\in \chi(t)$ forms a non-empty subtree of~$T$.
\end{itemize}
The sets $\chi(t)$ are called \emph{bags} of the decomposition~$\mathcal{T}$ and $\chi(t)$
is the bag associated with the tree node~$t$. The \emph{width} of a tree-decomposition
$(T,\chi)$ is the size of a largest bag minus~$1$.  A tree-decomposition of
minimum width is called \emph{optimal}. The \emph{treewidth} of a graph $G$,
denoted by $\tw(G)$, is the
width of an optimal tree decomposition of~$G$.
For the presentation of our dynamic programming algorithms, it is convenient to consider tree
decompositions in the following normal form~\cite{Kloks94}: A
tuple $(T,\chi)$ is a \emph{nice tree decomposition} of a graph $G$ if
$(T,\chi)$ is a tree decomposition of~$G$, the tree $T$ is rooted at
node $r$,
and each node of $T$ is of one of the following four types:
\begin{enumerate}
\item a \emph{leaf node}: a node $t$ having no children and $|\chi(t)|=1$;
\item a \emph{join node}: a node $t$ having exactly two children
  $t_1,t_2$, and $\chi(t)=\chi(t_1)=\chi(t_2)$;
\item an \emph{introduce node}: a node $t$ having exactly one child
  $t'$, and $\chi(t)=\chi(t')\cup \{v\}$ for a node $v$ of $G$;
\item a \emph{forget node}: a node $t$ having exactly one child $t'$,
  and $\chi(t)=\chi(t')\setminus \{v\}$ for a node $v$ of $G$.
\end{enumerate}
% Note that for every node $v$ of $G$ that does not occur in $\chi(r)$
% there is exactly one node $t_v$ with parent $t_v'$ such that $v\in
% \chi(t_v)$ and $v\notin \chi(t_v')$ ($t_v'$ is a forget node).  If $v$
% occurs in $\chi(r)$ we set $t_v=r$.  In both cases we say that
% \emph{$t_v$ is the last node of $v$}.

For convenience we will also assume that $\chi(r)=\emptyset$ for the root $r$
of~$T$. For $t \in V(T)$ we denote by $T_t$ the subtree of $T$ rooted
at $t$ and we write $\chi(T_t)$ for the set $\bigcup_{t'\in V(T_t)}\chi(t')$.
%We denote by $descendant_T(t)$ the set of descendants of $t$ in the tree $T$ with the reference to $T$ removed when it is clear from the context.
%We conclude the preliminaries with an useful propositions for treewidth.

\begin{PRO}[\cite{Kloks94,Bodlaender96,BodlaenderDDFLP16}]\label{pro:comp-tw}
  It is possible to compute an optimal (nice) tree-decomposition of an $n$-vertex graph $G$ with treewidth $k$ in time $k^{\bigO{k^3}}n$, and to compute a $5$-approximate one in time $2^{\bigO{k}}n$. Moreover, the number of nodes in the obtained tree
  decompositions is $\bigO{kn}$.
\end{PRO}
}

\subsection{Problem Parameterizations}
One advantage of the parameterized complexity paradigm is that it allows us to study the complexity of a problem w.r.t.~several parameterizations of interest/relevance. To provide a concise description of the parameters under consideration, we introduce the following terminology: We say that a $\bullet$ entry  at position $[i,j]$ in an incomplete matrix $\matM$ is \emph{covered} by row $i$ and by column $j$. In this paper, we study \textsc{RMC} and \textsc{DRMC} w.r.t.~the following parameterizations (see Figure~\ref{fig:parameterizations} for illustration):

\begin{itemize}
\vspace{-0.3cm}
\item $\col$: The minimum number of columns in the matrix $\matM$ covering all occurrences of $\bullet$ in $\matM$.
\vspace{-0.2cm}
\item $\row$: The minimum number of rows in the matrix $\matM$ covering all occurrences of $\bullet$ in $\matM$.
\vspace{-0.2cm}
\item $\comb$: The minimum value of $r+c$ such that there exist $r$
  rows and $c$ columns in $\matM$ with the property that each
  occurrence of $\bullet$ is covered one of these rows or columns.
%\vspace{-0.3cm}
%\nb{Maybe add number of question marks or non-question marks per row/column as a parameter. -- depends whether we get more results for that parameter -R}
\end{itemize}

\begin{figure}[tbhp]
\vspace{-0.1cm}
\begin{equation*}
\vspace{-0.1cm}
     \left(\begin{array}{cccccc}
     1 & 1 & 1 & 0 & \bullet & 1 \\
     0 & 0 & 1 & 0 & \bullet & 1 \\
     0 & \bullet & \bullet  & 0 & \bullet & \bullet \\
     1 & 1 & 0 & 1 & 0 & 1 \\
        \end{array}\right)
\end{equation*}
\caption{Illustration of the parameters $\col$, $\row$, and $\comb$ in an incomplete matrix. Here $\col=4$, $\row=3$, and $\comb=2$. \vspace{-0.1cm}}
\label{fig:parameterizations}
\end{figure}

We denote the parameter under consideration in brackets after the problem name (e.g., \textsc{DRMC}[\comb]).
As mentioned in Section~\ref{sec:intro}, both $p$-\textsc{RMC} and
$p$-\textsc{DRMC} are trivially in \FPT{} when parameterized by the number of missing entries, and hence this parameterization is not discussed further.

\lv{Given an incomplete matrix $\matM$, computing the parameter values for $\col$ and $\row$ is trivial. Furthermore, the parameter values satisfy $\comb\leq \row$ and $\comb \leq \col$. We show that the parameter value for $\comb$ can also be computed in polynomial time.
}
\sv{Given an incomplete matrix $\matM$, computing the parameter values for $\col$ and $\row$ is trivial. Furthermore, the parameter values satisfy $\comb\leq \row$ and $\comb \leq \col$. The parameter value for $\comb$ can also be computed in polynomial time by reducing the problem to finding a vertex cover in a bipartite graph:
}

\begin{PRO}
\label{pro:compcomb}
Given an incomplete matrix $\matM$ over $\GF(p)$, we can compute the parameter value for {\normalfont comb}, along with sets $R$ and $C$ of total cardinality {\normalfont comb} containing the indices of covering rows and columns, respectively, in time $\bigO{(n\cdot m)^{1.5}}$.
\end{PRO}
\lv{
\begin{proof}
We begin by constructing an auxiliary bipartite graph $G$ from $\matM$ as follows. For each row $i$ containing a $\bullet$, we create a vertex $v_i$ in $G$; similarly, for each column $j$ containing a $\bullet$, we create a vertex $w_j$. For each $\bullet$ that occurs at position $[i,j]$, we add an edge between $v_i$ and $w_j$.

We observe that if $R$ contains row indices and $C$ contains column indices which together cover all occurrences of $\bullet$, then $X=\{v_i \mid i\in R\}\cup \{w_j \mid j\in C\}$ is a vertex cover of $G$. Similarly, for each vertex cover $X$ of $G$, we can obtain a set $R=\{i \mid v_i\in X\}$ and a set $C=\{j \mid w_j\in X\}$ such that $R$ and $C$ cover all occurrences of $\bullet$ in $\matM$. Hence there is a one-to-one correspondence between vertex covers in $G$ and sets $R$ and $C$ which cover all $\bullet$ symbols in $\matM$, and in particular, the size of a minimum vertex cover in $G$ is equal to {\normalfont comb}. The lemma now follows by K\"onig's theorem and  Hopcroft-Karp's algorithm, which allow us to compute a minimum vertex cover in a bipartite graph $G$ in time $\bigO{|E(G)|\cdot \sqrt{|V(G)|}}$.
\end{proof}
}

\section{Rank Minimization}
\label{sec:rm}

In this section we present our results for \textsc{Bounded Rank Matrix Completion} under various parameterizations.

\subsection{Bounded Domain: Parameterization by \normalfont{row}}
As our first result, we present an algorithm for solving \textsc{$p$-RMC}[row]. This will serve as a gentle introduction to the techniques used in the more complex result for \textsc{$p$-RMC}[comb], and will also be used to give an \XP\ algorithm for \textsc{RMC}[row].

\begin{THE}
\label{the:mcrbrow}
\textsc{$p$-RMC}\emph{[\row{}]} is in \FPT{}.
%can be solved in time
%$\bigO{2^kp^{k^2}\cdot (m+kn)^{2.4}}$, where $k=\textup{\row{}}$, and hence is in \FPT{}.
\end{THE}
\lv{\begin{proof}}
\sv{\begin{proof}[Proof Sketch.]}
Let $R$ be the (minimum) set of rows that cover all occurrences of $\bullet$ in the input matrix $\matM$. Since the existence of a solution does
not change if we permute the rows of $\matM$, we permute the rows of $\matM$ so that the rows in $R$ have indices $1,\dots,k$.
%For simplicity, we will identify $R$ with the set $[k]$.
We now proceed in three steps.

For the first step, we will define the notion of \emph{signature}: A signature $S$ is a tuple $(I,D)$, where $I\subseteq R$ and $D$ is a mapping from $R\setminus I$ to $(I\rightarrow \GF(p))$. Intuitively, a signature $S$ specifies a subset $I$ of $R$ which is expected to be independent in $M[k+1:m,*]\cup I$ (i.e., adding the rows in $I$ to $M[k+1:m,*]$ is expected to increase the rank of $M[k+1:m,*]$ by $|I|$);
%\nb{Todo: Define what this $\cup$ means, what it means for $I$ to be independent in that thing we build there},
and for each remaining row of $R$, $S$ specifies how that row should depend on $I$. The latter is carried out using $D$: For each row in $R\setminus I$, $D$ provides a set of coefficients expressing the dependency of that row on the rows in $I$. Formally, we say that a matrix $\matM'$ that is compatible with the incomplete matrix $\matM$ \emph{matches} a signature $(I,D)$ if and only if, for each row (\ie, vector) $d\in R\setminus I$, there exist coefficients $a^d_{k+1},\dots, a^d_{m}\in \GF(p)$ such that $d=a^d_{k+1}\matM[k+1,*]+\dots+a^d_{m}\matM[m,*]+\sum_{i\in I}D(d)(i)\cdot i$.
The first step of the algorithm branches through all possible signatures $S$. Clearly, the number of distinct signatures is upper-bounded by $2^k\cdot p^{k^2}$.

\lv{
For the second step, we fix an enumerated signature $S$. The algorithm will verify whether $S$ is \emph{valid}, \ie, whether there exists a matrix $\matM'$ compatible with $\matM$ that matches $S$. To do so, the algorithm will construct a system of $|R\setminus I|$ equations over vectors of size $n$, and then transform this into a system $\Upsilon_S$ of $|R\setminus I|\cdot n$ equations over $\GF(p)$ (one equation for each vector coordinate). For each $d\in R\setminus I$, $\Upsilon_S$ contains the following variables:
\begin{itemize}
\item one variable for each coefficient $a^d_{k+1},\dots, a^d_{m}$, and
\item one variable for each occurrence of $\bullet$ in the rows of $R$.
\end{itemize}
}
\sv{
For the second step, we fix an enumerated signature $S$. The algorithm will verify whether $S$ is \emph{valid}, \ie, whether there exists a matrix $\matM'$ compatible with $\matM$ that matches $S$. To do so, the algorithm will construct a system of $|R\setminus I|$ equations over vectors of size $n$, and then transform this into a system $\Upsilon_S$ of $|R\setminus I|\cdot n$ equations over $\GF(p)$ (one equation for each vector coordinate). For each $d\in R\setminus I$, $\Upsilon_S$  contains one variable for each coefficient $a^d_{k+1},\dots, a^d_{m}$ and one variable for each occurrence of $\bullet$ in the rows of $R$.}
For instance, the first equation in $\Upsilon_S$ has the following form:
$d[1]=a^d_{k+1}\matM[k+1,1]+\dots+a^d_{m}\matM[m,1]+\sum_{i\in I}D(d)(i)\cdot i[1]$, where $a^d_{k+1},\dots, a^d_{m}$ are variables, and $d[1]$ as well as each $i[1]$ in the sum could be a variable or a fixed number. Crucially, $\Upsilon_S$ is a system of at most $(k\cdot n)$ \emph{linear} equations over $\GF(p)$ with at most $m+kn$ variables, and can be solved in time $\bigO{(m+kn)^{3}}$ by Gaussian elimination.
%~\cite{matrixcompbook}.
Constructing the equations takes time $\bigO{m\cdot n}$.

During the second step, the algorithm determines whether a signature $S$ is valid or not, and in the end, after going through all signatures, selects an arbitrary valid signature $S=(I,D)$ with minimum $|I|$.
For the final third step, the algorithm checks whether $|I|+\rk(\matM[k+1:m,*])\leq t$. We note that computing $\rk(\matM[k+1:m,*])$ can be carried out in time $\bigO{nm^{1.4}}$~\cite{IbarraMH82}. If the above inequality does not hold, the algorithm rejects; otherwise it recomputes a solution to $\Upsilon_S$ and outputs the matrix $\matM'$ obtained from $\matM$ by replacing each occurrence of $\bullet$ at position $[i,j]$ by the value of the variable $i[j]$ in the solution to $\Upsilon_S$.
The total running time is $\bigO{(2^k\cdot p^{k^2})\cdot ((m+kn)^{3}+nm^{1.4})}=\bigO{2^kp^{k^2}\cdot (m+kn)^{3}}$.
\lv{

To argue the correctness of the algorithm, consider a matrix $\matM'$ that the algorithm outputs. Obviously, $\matM'$ is consistent with $\matM$. Furthermore, $\matM'$ has rank at most $t$; indeed, the rank of $\matM[k+1:m,*]$ is at most $t-|I|$, and every row $d\in R\setminus I$ can be obtained as a linear combination of rows in $\matM[k+1:m,*]$ (using coefficients $a^d_{k+1},\dots,a^d_{m}$) and $I$ (using coefficients $D(d)$).

Conversely, assume that there exists a matrix $\matM'$ that is consistent with $\matM$ and that has rank at most $t$. Choose $\matM'$ to be of minimum rank over all matrices consistent with $\matM$. Consider the signature $S$ obtained (``reverse-engineered'') from $\matM'$ as follows. First, we choose a row-basis $B$ of $\matM'$ such that $|B\cap R|$ is minimized, and we set $I=R\cap B$. Now, each row in $R\setminus I$ can be obtained as a linear combination of rows in $B$ and, in particular, as a linear combination of the rows in $\matM'[k+1:m,*]$ and $I$. This can be expressed as a system of equations $\Upsilon'$, where, for each row $d\in R\setminus I$, we write $d=a^d_{k+1}\matM[k+1,*]+\dots+a^d_{m}\matM'[m,*]+\sum_{i\in I}D(d)(i)\cdot i$ and our variables are: $a^d_{k+1},\dots,a^d_{m}$ and, $\forall i\in I, D(d)(i)$. Let us fix an arbitrary solution to $\Upsilon'$ and use the values assigned to variables $D(d)(i)$, $\forall d\in R\setminus I, i\in I$, to define $D$. Observe that $S=(I,D)$ was chosen so that $\Upsilon_S$ is guaranteed to have a solution.

Next, we argue that $|I|+\rk(\matM[k+1:m,*])=\rk(\matM')$. Indeed, assume for a contradiction that there exists a row $r\in R\cap I$ which can be obtained as a linear combination of the rows in $I$ and in $\matM[k+1:m,*]$; then, we could replace $r$ in $B$ by a row $r'$ from $\matM[k+1:m,*]$ which would violate the minimality of $|B\cap R|$. So, $|I|+\rk(\matM[k+1:m,*])=\rk(\matM')$ which means that our algorithm is guaranteed to set $S$ as a valid branch, and hence, will either output a matrix compatible with $\matM$ which matches $S$, or a matrix compatible with $\matM$ which matches a different signature but has the same rank as $\matM'$.
}
\end{proof}
\sv{Since the the transpose of $\matM$ has the same rank as $\matM$, it
  follows immediately that \textsc{$p$-RMC}[\col{}] is in \FPT{}.}
\lv{Since the row-rank of a matrix $\matM$ is equal to its column-rank, the transpose of $\matM$ has the same rank as $\matM$. Hence:}

\begin{COR}
\label{cor:mcrbcol}
\textsc{$p$-RMC}\emph{[\col{}]} is in \FPT{}.
%can be solved in time
%$\bigO{2^kp^{k^2}\cdot (km+n)^{2.4}}$, where $k=\textup{\col{}}$, and hence is in \FPT.
\end{COR}

As a consequence of the running time of the algorithm given in the
proof of Theorem~\ref{the:mcrbrow}, we obtain:
\begin{COR}
\label{cor:XP}
\textsc{RMC}\emph{[\row{}]} and \textsc{RMC}\emph{[\col{}]} are in \XP.
\end{COR}

\subsection{Bounded Domain: Parameterization by \normalfont{comb}}
%\subsection{An Algorithm for \textsc{\normalfont{$p$-RMC}}{\normalfont [\comb]}}
In this subsection, we present a randomized fixed-parameter algorithm for \textsc{$p$-RMC}[comb] with constant one-sided error probability.
Before we proceed to the algorithm, we need to introduce some basic terminology related to systems of equations. Let $\Upsilon$ be a system of $\ell$ equations $\EQ_1$, $\EQ_2$,\dots, $\EQ_\ell$ over $\GF(p)$; we assume that the equations are simplified as much as possible. In particular, we assume that no equation contains two terms over the same set of variables such that the degree/exponent of each variable in both terms is the same.
Let $\EQ_i$ be a linear equation in $\Upsilon$, and let $x$ be a variable which occurs in $\EQ_i$ (with a non-zero coefficient).
Naturally, $\EQ_i$ can be transformed into an equivalent equation
$\EQ_{i,x}$, where $x$ is isolated, and we use $\Gamma_{i,x}$ to
denote the side of $\EQ_{i,x}$ not containing $x$, i.e., $\EQ_{i,x}$ is of the form $x=\Gamma_{i,x}$.
We say that $\Upsilon'$ is obtained from $\Upsilon$ by \emph{substitution of $x$ in $\EQ_i$} if $\Upsilon'$ is the system of equations obtained by:
\begin{enumerate}
\item computing $\EQ_{i,x}$ and in particular $\Gamma_{i,x}$ from $\EQ_i$;
\item setting $\Upsilon':=\Upsilon\setminus \{\EQ_i\}$; and
\item replacing $x$ with $\Gamma_{i,x}$ in every equation in $\Upsilon'$.
\end{enumerate}

Observe that $\Upsilon'$ has size $\bigO{n\cdot \ell}$, and can also be computed in time $\bigO{n\cdot \ell}$, where $n$ is the number of variables occurring in $\Upsilon$. Furthermore, any solution to $\Upsilon'$ can be transformed into a solution to $\Upsilon$ in linear time, and similarly any solution to $\Upsilon$ can be transformed into a solution to $\Upsilon'$ in linear time (\ie, $\Upsilon'$ and $\Upsilon$ are equivalent). Moreover, $\Upsilon'$ contains at least one fewer variable and one fewer equation than $\Upsilon$.

% The following proposition, which relies on earlier work by Courtois et al\ (\citeyear{CourtoisGMT02}), will be of crucial importance for our proof.
The following proposition is crucial for our proof, and is of
independent interest.
\begin{PRO}\label{pro:quadratic}
  Let $\Upsilon$ be a system of $\ell$ quadratic equations over
  $\GF(p)$.
  Then computing a solution for
  $\Upsilon$ is in \RFPT{} parameterized by $\ell$ and $p$, and in \RXP{} parameterized only by $\ell$.
\end{PRO}
\begin{proof}
  Let $n$ be the number of variables in $\Upsilon$. We distinguish two cases.
  If $n \geq \ell(\ell+3)/2$, then $\Upsilon$ can
  be solved in randomized time $\bigO{2^\ell n^3 \ell(\log p)^2}$~\cite{MiuraHT14}.
  Otherwise, $n < \ell(\ell+3)/2$, and we can solve $\Upsilon$ by a
  brute-force algorithm which enumerates (all of the) at most
  $p^{n}< p^{\ell(\ell+3)/2}$ assignments of values
  to the variables in $\Upsilon$. The proposition now follows by
  observing that the given algorithm runs in time $\bigO{2^\ell n^3 \ell(\log p)^2+p^{\ell(\ell+3)/2}\ell^2}$.
  % If $n\geq 2^{0.15\cdot k}\cdot(k+1)$, then $\Upsilon$ can
  % be solved in randomized polynomial time~\cite{CourtoisGMT02}. Otherwise, $n <
  % 2^{0.15\cdot k}\cdot(k+1)$, and we can solve $\Upsilon$ by a
  % brute-force algorithm which enumerates (all of the) at most
  % $p^{n}\leq p^{2^{0.15\cdot k}\cdot(k+1)}$ assignments of values
  % to the variables in $\Upsilon$. The proposition now follows by
  % observing that the given algorithm requires time at most $\bigO{n^2\cdot 2^{\bigO{k}}+p^{2^{0.15\cdot k}\cdot(k+1)}}$.
  % It can be verified that the running time of
  % Proposition~\ref{pro:quadratic} with $k'$ equations is upper-bounded
  % by $n^2\cdot 2^{\bigoh(k')}$ (in the case where $n$ is large) plus
  % $p^{2^{0.15\cdot k'}\cdot(k'+1)}+n$ (in the case where we
  % brute-force the assignments).
\end{proof}

\begin{THE}
\label{thm:rbcomb}
\textsc{$p$-RMC}\emph{[\comb]} is in \RFPT.
\end{THE}

\lv{\begin{proof}}
\sv{\begin{proof}[Proof Sketch.]}
We begin by using Proposition~\ref{pro:compcomb} to compute the sets $R$ and $C$ containing the indices of the covering rows and columns, respectively; let $|R|=r$ and $|C|=c$, and recall that the parameter value is $k=r+c$.
Since the existence of a solution for \textsc{$p$-RMC} does not change if we permute rows and columns of $\matM$, we permute the rows of $\matM$ so that the rows in $R$ have indices $1,\dots,r$, and subsequently, we permute the columns of $\matM$ so that the columns in $C$ have indices $1,\dots,c$.

%\nb{picture: matrix}

Before we proceed, let us give a high-level overview of our strategy. The core idea is to branch over signatures, which will be defined in a similar way to those in Theorem~\ref{the:mcrbrow}. These signatures will capture information about the dependencies among the rows in $R$ and columns in $C$; one crucial difference is that for columns, we will focus only on dependencies in the submatrix $\matM[r+1:m,*]$\lv{
(the reason will become clear later, when we argue correctness)}.
In each branch, we arrive at a system of equations that needs to be
solved in order to determine whether the signatures are valid. Unlike
Theorem~\ref{the:mcrbrow}, here the obtained system of equations will
contain non-linear (but quadratic) terms, and hence solving the system is far from being trivial.
Once we determine which signatures are valid, we choose one that minimizes the total rank.

For the first step, let us define the notion of signature that will be used in this proof. A \emph{signature} $S$ is a tuple $(I_R, D_R, I_C, D_C)$ where:
\sv{1. $I_R\subseteq R$; 2. $D_R$ is a mapping from $R\setminus I_R$ to $(I_R\rightarrow \GF(p))$; 3. $I_C\subseteq C$; and 4. $D_C$ is a mapping from $C\setminus I_C$ to $(I_C\rightarrow \GF(p))$.}
\lv{
\begin{enumerate}
\item $I_R\subseteq R$;
\item $D_R$ is a mapping from $R\setminus I_R$ to $(I_R\rightarrow \GF(p))$;
\item $I_C\subseteq C$; and
\item $D_C$ is a mapping from $C\setminus I_C$ to $(I_C\rightarrow \GF(p))$.
\end{enumerate}}

We say that a matrix $\matM'$ compatible with the incomplete matrix $\matM$ \emph{matches} a signature $(I_R, D_R, I_C, D_C)$ if:
\begin{itemize}
\item  for each row $d\in R\setminus I_R$, there exist coefficients $a^d_{r+1},\dots, a^d_{m}\in \GF(p)$ such that $d=a^d_{r+1}\matM'[r+1,*]+\dots+a^d_{m}\matM'[m,*]+\sum_{i\in I_R}D_R(d)(i)\cdot i$; and
\item for each column $h\in C\setminus I_c$, there exist coefficients $b^h_{c+1},\dots, b^h_{n}\in \GF(p)$ such that $h[r+1:m]=b^h_{c+1}\matM'[r+1:m,c]+\dots+b^h_{n}\matM'[r+1:m,n]+\sum_{i\in I_C}D_C(h)(i)\cdot i[r+1:m]$.
\end{itemize}

The number of distinct signatures is upper-bounded by $2^r\cdot p^{r^2}\cdot 2^c\cdot p^{c^2}\leq 2^k\cdot p^{k^2}$, and the first step of the algorithm branches over all possible signatures $S$.
\sv{In the second step, for each enumerated signature $S$, we check whether $S$ is valid (\ie, whether there exists a matrix $\matM'$, compatible with the incomplete $\matM$, that matches $S$) in a similar fashion as in the proof of Theorem~\ref{the:mcrbrow}. Here, this results in a system $\Upsilon_S$ of $|R\setminus I_R|\cdot n+|C\setminus I_C|\cdot (m-r)$ equations which check the dependencies for rows in $R\setminus I_R$ and columns in $C\setminus I_C$. For instance, the first equation in $\Upsilon_S$ for some $d\in R\setminus I_R$ has the following form:
$d[1]=a^d_{r+1}\matM[r+1,1]+\dots+a^d_{m}\matM[m,1]+\sum_{i\in I_R}D_R(d)(i)\cdot i[1]$, where $a^d_{r+1},\dots, a^d_{m}$ are variables, $D_R(d)(i)$ is a number, and all other occurrences are either variables or numbers. Similarly, the second equation in $\Upsilon_S$ for some $h\in C\setminus I_C$ has the following form: $h[r+2]=b^h_{c+1}\matM[r+2,c+1]+\dots+b^h_{n}\matM[r+2,n]+\sum_{i\in I_C}D_C(d)(i)\cdot i[r+2]$, where $b^h_{c+1},\dots, b^h_{n}$ are variables, $D_C(d)(i)$ is a number, and all other occurrences are either variables or numbers.

Next, observe that the only equations in $\Upsilon_S$ that may contain
non-linear terms are those for $d[j]$, where $j\leq c$, and in
particular $\Upsilon_S$ contains at most $k^2$ equations with
non-linear terms ($k$ equations for at most $k$ vectors $d$ in
$R\setminus I_R$). We will now use substitutions to simplify
$\Upsilon_S$ by removing all linear equations; specifically, at each
step we select an arbitrary linear equation $\EQ_i$ containing a
variable $x$, apply substitution of $x$ in $\EQ_i$ to construct a new
system of equations with one fewer equation, and simplify all
equations in the new system. If at any point we reach a system of
equations that contains an invalid equation (\eg, 2=5), then
$\Upsilon_S$ does not have a solution, and we discard the
corresponding branch. Otherwise, after at most $|R\setminus I_R|\cdot
n+|C\setminus I_C|\cdot (m-r)\in \bigO{kn+km}$ substitutions, we
obtain a system of at most $k^2$ quadratic equations $\Psi_S$ such
that any solution to $\Psi_S$ can be transformed into a solution to
$\Upsilon_S$ in time $\bigO{kn+km}$. We can now apply
Proposition~\ref{pro:quadratic} to solve $\Psi_S$ and mark $S$ as a
valid signature if $\Psi_S$ has a solution.

% continue here
% Our goal now is to solve $\Psi_S$, and to do so, we distinguish two cases. Let $n'$ and $m'$ be the number of variables and the number of equations in $\Psi_S$, respectively, and recall that $m'\leq k^2$. If $n'\geq 2^{0.15\cdot m'}\cdot(m'+1)$, then $\Upsilon_S$ can be solved in polynomial time~\cite{CourtoisGMT02}. Otherwise, $n' < 2^{0.15\cdot m'}\cdot(m'+1)$, and we can solve $\Psi_S$ by a brute-force algorithm which enumerates (all of the) at most $p^{n'}\leq p^{2^{0.15\cdot m'}\cdot(m'+1)}$ assignments of values to the variables in $\Psi_S$. Hence, $\Psi_S$ can be solved in \FPT{}-time. If $\Psi_S$ has a solution, then we mark $S$ as a valid signature.

After all signatures have been processed, the algorithm selects a
valid signature $S=(I,D)$ that has the minimum value of $|I_R|+|I_C|$,
checks whether $|I_R|+|I_C|+\rk(\matM[r+1:m,c:1+n])\leq t$, and either
uses this to construct a solution (similarly to the proof of
Theorem~\ref{the:mcrbrow}), or outputs ``no''.
The theorem now follows by observing that the total running time of
the algorithm is obtained by combining the branching factor of branching over all signatures ($\bigO{2^k\cdot p^{k^2}}$) with the run-time of Proposition~\ref{pro:quadratic} for $k^2$ many quadratic equations ($\bigO{3^{k^2} n^3 (\log p)^2+p^{k^4}}$). In particular, we obtain a running time of $\bigO{3^{k^2} \cdot p^{k^4} \cdot n^3}$.\end{proof}
}
%Crucially, for all $j>c$ and all $j'$, all equations for $h[j']$ and the equations for $d[j]$ contain only linear terms; however, for $j\in [c]$ the equations for $d[j]$ may also contain non-linear terms (in particular, $a^d_{r+1},\dots,a^d_{m}$ are variables and $\matM[r+1,j],\dots,\matM[m,j]$ can contain $\bullet$ symbols, which correspond to variables in the equations).
\lv{
For the second step, fix an enumerated signature $S$. The algorithm will verify whether $S$ is \emph{valid}, \ie, whether there exists a matrix $\matM'$, compatible with the incomplete $\matM$, that matches $S$. To do so, the algorithm will construct a system of $|R\setminus I_R|$ equations over vectors of size $n$ and of $|C\setminus I_C|$ equations over vectors of size $m-r$, and then transform this into a system $\Upsilon_S$ of $|R\setminus I_R|\cdot n+|C\setminus I_C|\cdot (m-r)$ equations over $\GF(p)$ (one equation for each vector coordinate). For each $d\in R\setminus I_R$, $\Upsilon_S$ contains the following variables:
\begin{itemize}
\item one variable for each $a^d_{r+1},\dots, a^d_{m}$, and
\item one variable for each occurrence of $\bullet$.
\end{itemize}
For instance, the first equation in $\Upsilon_S$ for some $d\in R\setminus I_R$ has the following form:
$d[1]=a^d_{r+1}\matM[r+1,1]+\dots+a^d_{m}\matM[m,1]+\sum_{i\in I_R}D_R(d)(i)\cdot i[1]$, where $a^d_{r+1},\dots, a^d_{m}$ are variables, $D_R(d)(i)$ is a number, and all other occurrences are either variables or numbers. Crucially, for all $j>c$, the equations for $d[j]$ defined above contain only linear terms; however, for $j\in [c]$ these equations may also contain non-linear terms (in particular, $a^d_{r+1},\dots,a^d_{m}$ are variables and $\matM[r+1,j],\dots,\matM[m,j]$ can contain $\bullet$ symbols, which correspond to variables in the equations).
For $z\in [m]$ and $y\in [n]$, if an element $\matM[z,y]$ contains $\bullet$, then we will denote the corresponding variable used in the equations as $x_{z,y}$.

Next, for each $h\in C\setminus I_C$, $\Upsilon_S$ contains the following variables:
\begin{itemize}
\item one variable for each $b^h_{c+1},\dots, b^h_{n}$, and
\item one variable for each occurrence of $\bullet$.
\end{itemize}
For instance, the second equation in $\Upsilon_S$ for some $h\in C\setminus I_C$ has the following form:
$h[r+2]=b^h_{c+1}\matM[r+2,c+1]+\dots+b^h_{n}\matM[r+2,n]+\sum_{i\in I_C}D_C(d)(i)\cdot i[r+2]$, where $b^h_{c+1},\dots, b^h_{n}$ are variables, $D_C(d)(i)$ is a number, and all other occurrences are either variables or numbers. Observe that all of these equations for $h$ are linear, since the submatrix $\matM[r+1:m,c+1:n]$ contains no $\bullet$ symbols.
%\nb{Maybe give example 3x3 or 4x4 matrix and equations? -R}

This completes the definition of our system of equations $\Upsilon_S$. Recall that the only equations in $\Upsilon_S$ that may contain non-linear terms are those for $d[j]$ when $j\leq c$, and in particular $\Upsilon_S$ contains at most $k^2$ equations with non-linear terms ($k$ equations for at most $k$ vectors $d$ in $R\setminus I_R$). We will now use substitutions to simplify $\Upsilon_S$ by removing all linear equations; specifically, at each step we select an arbitrary linear equation $\EQ_i$ containing a variable $x$, apply substitution of $x$ in $\EQ_i$ to construct a new system of equations with one fewer equation, and simplify all equations in the new system. If at any point we reach a system of equations which contains an invalid equation (\eg, 2=5), then $\Upsilon_S$ does not have a solution, and we discard the corresponding branch. Otherwise, after at most $|R\setminus I_R|\cdot n+|C\setminus I_C|\cdot (m-r)\in \bigO{kn+km}$ substitutions we obtain a system of at most $k^2$ quadratic equations $\Psi_S$ such that any solution to $\Psi_S$ can be transformed into a solution to $\Upsilon_S$ in time at most $\bigO{kn+km}$.
We can now apply
Proposition~\ref{pro:quadratic} to solve $\Psi_S$ and mark $S$ as a
valid signature if $\Psi_S$ has a solution.
% Our goal now is to solve $\Psi_S$, and to do so, we distinguish two cases. Let $n'$ and $m'$ be the number of variables and the number of equations in $\Psi_S$, respectively, and recall that $m'\leq k^2$. If $n'\geq 2^{0.15\cdot m'}\cdot(m'+1)$, then $\Upsilon_S$ can be solved in polynomial time~\cite{CourtoisGMT02}. Otherwise, $n' < 2^{0.15\cdot m'}\cdot(m'+1)$, and we can solve $\Psi_S$ by a brute-force algorithm which enumerates all the at most $p^{n'}\leq p^{2^{0.15\cdot m'}\cdot(m'+1)}$ assignments of values to the variables in $\Psi_S$. Hence, $\Psi_S$ can be solved in \FPT{}-time. If $\Psi_S$ has a solution, then we mark $S$ as a valid signature.

After all signatures have been processed, in the third---and final---step we select a valid signature $S=(I,D)$ that has the minimum value of $|I_R|+|I_C|$. The algorithm will then check whether $|I_R|+|I_C|+\rk(\matM[r+1:m,c:1+n])\leq t$. If this not the case, the algorithm rejects the instance. Otherwise, the algorithm recomputes a solution to $\Upsilon_S$, and outputs the matrix $\matM'$ obtained from $\matM$ by replacing each occurrence of $\bullet$ at position $[i,j]$ in $\matM'$ by $x_{i,j}$.

%It can be verified that the running time of Proposition~\ref{pro:quadratic} with $k'$ equations is upper-bounded by $n^2\cdot 2^{\bigO{k'}}$ (in the case where $n$ is large) plus $p^{2^{0.15\cdot k'}\cdot(k'+1)}+n$ (in the %case where we brute-force the assignments). Hence the total running time of our algorithm is upper-bounded by $p^{2^{\bigO{k^2}}}\cdot n^2$.

We now proceed to proving the correctness of the algorithm. We do so by proving the following two claims:

%First, we show that if a signature is selected as valid, then we can use it to construct a matrix $\matM'$ compatible with $\matM$ such that the rank of $\matM'$ has the desired upper bound.

\begin{CLM}
If there exists a signature $S=(I_R, D_R, I_C, D_C)$ for $\matM$ such that $|I_R|+|I_C|+\rk(\matM[r+1:m,c:1+n])\leq t$, then there exists a matrix $\matM'$ compatible with $\matM$ such that $\rk(\matM')\leq |I_R|+|I_C|+\rk(\matM[r+1:m,c:1+n])$. In particular, if $S$ is marked as valid by the algorithm, then the algorithm outputs a matrix $\matM'$ satisfying the above.
\end{CLM}

\begin{proof}[Proof of Claim]
\renewcommand{\qedsymbol}{$\blacksquare$}
Since $S$ is valid, the system of equations $\Upsilon_S$ has a solution; fix one such solution. Consider the matrix $\matM'$ obtained from $\matM$ by replacing each occurrence of $\bullet$ at position $[i,j]$ by the value of $x_{i,j}$ from the selected solution to $\Upsilon_S$. Then the solution to $\Upsilon_S$ guarantees that each row in $R\setminus I_R$ can be obtained as a linear combination of rows in $\matM'\setminus (R\setminus I_R)$, and hence deleting all rows in $R\setminus I_R$ will result in a matrix $\matM_1'$ such that $\rk(\matM')=\rk(\matM_1')$.

Next, consider the matrix $\matM'[r+1:m,*]$ which is obtained by removing all rows in $I_R$ from $\matM_1'$; clearly, this operation decreases the rank by at most $|I_R|$, and hence $\rk(\matM'[r+1:m,*])\leq \rk(\matM_1')\leq \rk(\matM'[r+1:m,*])+|I_R|$.

Third, consider the matrix $\matM_2'$ obtained from $\matM'[r+1:m,*]$ by removing all columns in $C\setminus I_C$. The solution to $\Upsilon_S$ guarantees that each removed column can be obtained as a linear combination of columns in $\matM'[r+1:m,*]\setminus (C\setminus I_C)$, and hence, $\rk(\matM'[r+1:m,*)=\rk(\matM_2')$. Finally, we consider the matrix $\matM'[r+1:m,c+1:n]=\matM[r+1:m,c+1:n]$ which is obtained by removing all columns in $I_C$ from $\matM_2'$. Clearly, removing $|I_C|$ columns decreases the rank by at most $|I_C|$, and hence $\rk(\matM[r+1:m,c+1:n])\leq \rk(\matM_2')\leq \rk(\matM[r+1:m,c+1:n])+|I_C|$. Putting the above inequalities together, we get $\rk(\matM')\leq \rk(\matM'[r+1:m,*])+|I_R|\leq \rk(\matM[r+1:m,c+1:n])+|I_C|+|I_R|$.
\end{proof}

%The last claim we need to prove in order to establish the theorem is that if a solution exists, then our algorithm will find such a solution.

\begin{CLM}
If there exists a matrix $\matM'$ compatible with $\matM$ such that $\rk(\matM')\leq t$, then there exists a valid signature $S=(I_R, D_R, I_C, D_C)$ such that $|I_R|+|I_C|+\rk(\matM[r+1:m,c:1+n])\leq t$.
\end{CLM}

\begin{proof}[Proof of Claim]
\renewcommand{\qedsymbol}{$\blacksquare$}
Consider the following iterative procedure that creates a set $I_R$ from the hypothetical matrix $\matM'$. Check, for each row $r\in R$, whether $R$ can be obtained as a linear combination of all other rows in $\matM'$, which can be done by solving a system of linear equations; if this is the case, remove $r$ from $\matM'$ and restart from any row in $R$ that remains in $\matM'$. In the end, we obtain a submatrix $\matM'_R$ of $\matM'$ which only contains those rows in $R$ that cannot be obtained as a linear combination of all other rows in $\matM'_R$; let $I_R$ be the set of rows in $R$ that remain in $\matM'_R$. Furthermore, since each row $r'\in R\setminus I_R$ can be obtained as a linear combination of rows in $\matM'_R$, for each such $r'$ we compute a set of coefficients $\tau_{r'}$ that can be used to obtain $r'$ and store those coefficients corresponding to $I_R$ in $D_R$. For instance, if row $r'\in R\setminus I_R$ can be obtained by an additive term containing $1$ times row  $u\in I_R$, then we set $D_R(r')=(u\mapsto 1)$.

At this point, we have identified $I_R$ and $D_R$. Next, we turn our attention to the submatrix $\matM'[r+1:m,*]$, where we proceed similarly but for columns. In particular, for each column $c\in C$ restricted to $\matM'[r+1:m,*]$, we check whether $c$ can be obtained as a linear combination of  all other columns in $\matM'[r+1:m,*]$, and if the answer is positive then we remove $c$ from $\matM'[r+1:m,*]$ and restart from any column in $C$ that remains in $\matM'[r+1:m,*]$. This results in a new submatrix $\matM'_C$ of $\matM'[r+1:m,*]$, and those columns of $C$ that remain in $\matM'_C$ are stored in $I_C$. Then, for each column in $c'\in C\setminus I_C$, we compute a set of coefficients $\tau_{c'}$ that can be used to obtain that column and store the values of the coefficients that correspond to $I_C$ in $D_C$, analogously as we did for the rows.

At this point, we have obtained a signature $S$. The validity of $S$ follows from its construction. Indeed, to solve $\Upsilon_S$, we can set each variable $x_{i,j}$ representing the value of a $\bullet$ symbol at $\matM[i,j]$ to $\matM'[i,j]$, and all other variables will capture the coefficients that were stored in $\tau_{r'}$ and $\tau_{c'}$ for a row $r'$ or a column $c'$, respectively.

Finally, we argue that $|I_R|+|I_C|+\rk(\matM[r+1:m,c:1+n])\leq t$. Since $\matM'_R$ was obtained from $\matM'$ only by deleting linearly dependent rows, $\rk(\matM')=\rk(\matM'_R)$. Furthermore, since $\matM'[r+1:m,*]$ can be obtained by deleting $|I_R|$ rows from $\matM'_R$, and all of these deleted rows are linearly independent of all other rows in $\matM'_R$, we obtain $\rk(\matM'[r+1:m,*])=\rk(\matM'_R)-|I_R|$. By repeating the above arguments,
we see that $\rk(\matM'[r+1:m,*])=\rk(\matM'_C)$ and $\rk(\matM'[r+1:m,c+1:n])=\rk(\matM'_C)-|I_C|$. Recall that $\matM'[r+1:m,c+1:n]=\matM[r+1:m,c+1:n]$. Putting the above together, we obtain $\rk(\matM'[r+1:m,c+1:n])+|I_C|=\rk(\matM'[r+1:m,*])$, and $\rk(\matM'[r+1:m,c+1:n])+|I_C|+|I_R|=\rk(\matM')\leq t$.
\end{proof}

Finally, the total running time of the algorithm is obtained by combining the branching factor of branching over all signatures ($\bigO{2^k\cdot p^{k^2}}$) with the run-time of Proposition~\ref{pro:quadratic} for $k^2$ many quadratic equations ($\bigO{3^{k^2} n^3 (\log p)^2+p^{k^4}}$). We obtain a running time of $\bigO{3^{k^2} \cdot p^{k^4} \cdot n^3}$.
\end{proof}
}

As a consequence of the running time of the algorithm given in the
proof of Theorem~\ref{thm:rbcomb}, we obtain:
\begin{COR}
\label{cor:combxp}
\textsc{RMC}\emph{[\comb]} is in \RXP.
\end{COR}

% commands used in DRM section
\newcommand{\comG}{G}

\newcommand{\PIC}{\textsc{CC}}
\newcommand{\PIT}{\textsc{PIT}}
\newcommand{\TSATT}{\textsc{3-Sat-2}}

\section{\textsc{Bounded Distinct Row Matrix Completion}}
\label{sec:fr}

Let $(p,\matM,t)$ be an instance of \uDRM{}.
We say that two rows of $\matM$ are \emph{compatible}
if whenever the two rows differ at some entry then one of the rows has
a $\bullet$ at that entry.
The \emph{compatibility graph} of $\matM$, denoted by $\comG(\matM)$, is the undirected graph
whose vertices correspond to the row indices of $\matM$ and in which there is an edge between two
vertices if and only if their two corresponding rows are compatible. See Figure~\ref{fig:compatibilitygraph}
for an illustration.

\begin{figure}[htbp]
 \footnotesize
 \vspace{-0.1cm}
\begin{subfigure}{.4\textwidth}
\begin{equation*}
     \left(\begin{array}{cccccc}
     1 & \bullet & 0 & \bullet & \bullet & 1 \\
     1 & 0 & 0 & 1 & \bullet & \bullet \\
     1 & 0 & \bullet & 1 & 0 & 1 \\
     1 & 0 & 1 & 1 & 0 & \bullet \\
     1 & 0 & 1 & 1 & 0 & 0 \\
        \end{array}\right)
\end{equation*}
\end{subfigure}%
\hspace{0.5cm}\begin{subfigure}{.4\textwidth}
 \begin{tikzpicture}

\tikzstyle{every node}=[]
\tikzstyle{gn}=[circle, inner sep=2pt,draw]
\tikzstyle{dots}=[circle, inner sep=1pt,draw]
\tikzstyle{every edge}=[draw, line width=1.5pt]

\draw
node[gn, label=above:$1$] (v1) {}
node[gn, above right of=v1, label=right:$2$] (v2) {}
node[gn, below right of=v2, label=above:$3$] (v3) {}
node[gn, right of=v3, label=right:$4$] (v4) {}
node[gn, above of=v4, label=right:$5$] (v5) {}

(v1) edge (v2)
(v1) edge (v3)
(v2) edge (v3)
(v3) edge (v4)
(v4) edge (v5)
;

% \draw (0,0) node[gn, label=90:{$1$}] (v1) {};
% \draw (2,0) node[gn, label=90:{$3$}] (v2) {};
% \draw (1,1) node[gn, label=90:{$2$}] (v3) {};
% \draw (3,0) node[gn, label=90:{$4$}] (v4) {};
% \draw (3,-1) node[gn, label=0:{$5$}] (v5) {};

% \draw (v1) edge (v2);
% \draw (v1) edge (v3);
% \draw (v2) edge (v3);
% \draw (v2) edge (v4);
% \draw (v4) edge (v5);

\end{tikzpicture}
\end{subfigure}
  \vspace{-0.1cm}
\caption{Illustration of a matrix and its compatibility graph. The vertex label indicates the corresponding row number. \vspace{-0.1cm}}
\label{fig:compatibilitygraph}
\end{figure}
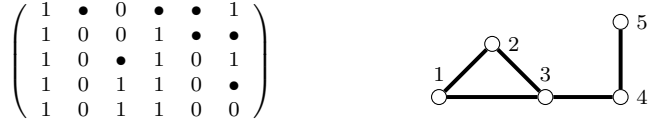

We start by showing that \uDRM{} (and therefore \bDRM{}) can be reduced to the
\textsc{Clique Cover} problem, which is defined as
follows.
\pbDef{\textsc{Clique Cover} (\PIC{})}{An undirected graph $G$ and
  an integer $k$.}{Find a partition of $V(G)$ into
  at most $k$ cliques, or output that no such partition exists.}
\begin{LEM}\label{lem:drm-cc}
  An instance $\III=(p,\matM,t)$ of \uDRM{} has a solution if and only if the
  instance $\III'=(\comG(\matM),t)$ of \PIC{} does. Moreover, a solution for $\III'$ can
  obtained in polynomial-time from a solution for $\III$ and vice versa.
\end{LEM}
\sv{\begin{proof}[Proof Sketch]
    The lemma follows immediately from the observation that a set $R$ of
    rows in $\matM$ can be made identical if and only if
    $\comG(\matM)[R]$ is a clique.
  \end{proof}
  }
\lv{\begin{proof}
  Let $\matM'$ be a solution
  for $\III$ and let $\PPP$ be the partition of the indices of the
  rows of $\matM'$ such that two row-indices $r$ and $r'$ belong to
  the same set in $\PPP$ if and only if
  $\matM'[r,*]=\matM'[r',*]$. Then $\PPP$ is also a solution for
  $\III'$, since $\comG[P]$ is a clique for every $P \in \PPP$.

  Conversely, let $\PPP$ be a solution for
  $\III'$. We claim that there is a solution $\matM'$ for $\III$ such
  that $\matM'[r,*]=\matM'[r',*]$ if and only if $r$ and $r'$ are
  contained in the same set of $\PPP$. Towards showing this, consider a set $P \in \PPP$ and a column
  index $c$ of $\matM$, and let $\MEnt(\matM[P,c])$ be the set of all values
  occurring in $\matM[P,c]$. Then $|\MEnt(\matM[P,c])\setminus \{\bullet\}|\leq 1$, that is, all entries of $\matM[P,c]$ that are not
  $\bullet$ are equal; otherwise, $\comG[P]$ would not be a
  clique. Consequently, by replacing every $\bullet$ occurring in
  $\matM[P,c]$ with the unique value in $\MEnt(\matM[P,c])\setminus \{\bullet\}$ if
  $\MEnt(\matM[P,c])\setminus \{\bullet\}\neq \emptyset$, and with an
  arbitrary value otherwise, and by doing so for every column index
  $c$ and every $P \in \PPP$, we obtain
  the desired solution $\matM'$ for $\III$.
\end{proof}
}
% Some notes on \PIC{} parameterized by treewidth:
% \begin{itemize}
% \item shall we mention that this implies an fpt-algorithm for \textsc{Coloring} parameterized by
%   the treewidth of the complement graph?
% \item I think previously the only related result was the \XP{}
%   algorithm parameterized by cliquewidth (in the MSO-partitioning
%   paper), however, for cliquewidth it turns out to be \W{1}\hy hard
%   (because \textsc{Coloring} is \W{1}\hy hard parameterized by
%   cliquewidth and cliquewidth is not changed under complement).
% \end{itemize}

\begin{THE}\label{lem:cc-tw}
  \PIC{} is in \FPT{} when parameterized by the treewidth of the input graph.
\end{THE}
\sv{\begin{proof}[Proof Sketch]
    We show the theorem via a standard dynamic programming
    algorithm on a tree-decomposition of the input
    graph~\cite{BodlaenderKoster08}. Namely after computing an optimal
    tree-decomposition $(T,\chi)$ of the input graph $G$, which can be
    achieved in \FPT-time w.r.t. the treewidth of
    $G$~\cite{Kloks94,Bodlaender96,BodlaenderDDFLP16}, we compute a set
    $\RRR(t)$ of tuples via a bottom-up dynamic programming algorithm for
    every $t \in V(T)$. In our case $(\PPP,c) \in \RRR(t)$ if and
    only if $\PPP$ is a partition
    of $G[\chi(t)]$ into cliques and $c$ is the minimum number such that
    $G[\chi(T_t)]$ has a partition
    $\PPP'$ into $c$ cliques with $\PPP=\SB P' \cap \chi(t) \SM
    P'\in\PPP'\SE\setminus \{\emptyset\}$.
\end{proof}}
\lv{\begin{proof}
  Let $\III=(G,k)$ be an instance of \PIC{}. We will show the Theorem
  using a standard dynamic programming algorithm on a
  tree-decomposition of $G$. Because of Proposition~\ref{pro:comp-tw}
  we can assume that we are given a nice tree-decomposition $(T,\chi)$
  of $G$ of width $\omega$. For every node $t \in V(T)$ we will
  compute the set $\RRR(t)$ of records containing all pairs
  $(\PPP,c)$, where $\PPP$ is a partition of $\chi(t)$ into cliques,
  i.e., for every $P \in \PPP$ the graph $G[P]$ is a clique, and $c$
  is the minimum integer such that $G[\chi(T_t)]$ has a partition
  $\PPP'$ into $c$ cliques with $\PPP=\SB P' \cap \chi(t) \SM
  P'\in\PPP'\SE\setminus \{\emptyset\}$. Note that $\III$ has a
  solution if and only if $\RRR(r)$ contains a record $(\emptyset,c)$
  with $c \leq k$, where $r$ is the root of $(T,\chi)$. It hence
  suffices to show how to compute the set of records for the four
  different types of nodes of a nice tree-decomposition.

  Let $l$ be a leaf node of $T$ with $\chi(l)=\{v\}$. Then
  $\RRR(l):=\{(\{\{v\}\},1)\}$. Note the $\RRR(l)$ can be computed in
  constant time.

  Let $t$ be an introduce node of $T$ with child $t'$ and
  $\chi(t)=\chi(t')\cup \{v\}$. Then $\RRR(t)$ can be obtained from
  $\RRR(t')$ as follows. For every $(\PPP',c') \in
  \RRR(t')$ and every $P' \in \PPP'$ such that $G[P'\cup\{v\}]$ is a
  clique, we add the record $((\PPP'\setminus \{P'\})\cup
  \{P'\cup\{v\}\},c')$ to $\RRR(t)$. Moreover, for every $(\PPP',c') \in
  \RRR(t')$, we add the record $(\PPP'\cup \{\{v\}\},c'+1)$ to $\RRR(t)$.
  Note that $\RRR(t)$ can be computed in time $\bigO{|\RRR(t')|\omega^2}$.

  Let $t$ be a forget node of $T$ with child $t'$ and
  $\chi(t)  \cup \{v\}=\chi(t')$. Then $\RRR(t)$ consists of all
  records $(\PPP,c)$ such that $c$ is the minimum integer such that
  there is a record $(\PPP',c) \in \RRR(t')$ and a set $P' \in \PPP'$
  with $v \in P'$ and $(\PPP'\setminus P')\cup (P'\setminus
  \{v\})=\PPP$; if no such record exists $(\PPP,c)$ is not in $\RRR(t)$.
  Note that $\RRR(t)$ can be computed in time
  $\bigO{|\RRR(t')|\omega^2}$.

  Let $t$ be a join node with children $t_1$ and $t_2$. Then $\RRR(t)$
  contains all records $(\PPP,c)$ such that there are integers $c_1$
  and $c_2$ with $c_1+c_2-|\PPP|=c$ and $(\PPP,c_1) \in \RRR(t_1)$ and
  $(\PPP,c_2) \in \RRR(t_2)$. Note that $\RRR(t)$ can be computed in
  time $\bigO{(|\RRR(t_1)|+|\RRR(t_2)|)\omega}$ (assuming that the records are
  kept in an ordered manner).

  The total run-time of the algorithm is then the number of nodes of
  $T$, i.e., $\bigO{\omega |V(G)|}$, times the maximum time required
  at any of the four types of nodes, i.e., $\bigO{|\RRR(t)|\omega^2}$,
  which because $|\RRR(t)| \leq \omega!$ is at most
  $\bigO{\omega!\omega^3|V(G)|}$.
  % old version using MSO
  \comment{
  Let $\III=(G,k)$ be an instance of \PIC{}. We will show the lemma
  with the help of Proposition~\ref{prop:MSO}\nb{Iyad: Missing proposition.}, \ie, by providing an
  \MSO{} formula $\Phi(X)$ (of constant length) with a free
  variable $X$ such that
  $\III$ has a solution if and only if there is a set $X$ of at most
  $k$ vertices of $G$ satisfying $G \models \Phi(X)$.

  The main idea behind the construction of $\Phi(X)$ is to
  guess a set $Y$ of edges of $G$ satisfying:
  \begin{itemize}
  \item[(P1)] every component $C$ of $G[Y]$ is a clique; and
  \item[(P2)] $X$ contains at least one vertex from every component of $G[Y]$.
  \end{itemize}
  Clearly, $G$ has a partitioning into $k$ cliques if and only if
  there is a set $X$ with $|X|\leq k$ such that there exists a subset
  $Y$ of $E(G)$ satisfying (P1) and (P2). It remains to show
  how to enforce (P1) and (P2) using $\Phi(X)$.

  Informally, (P1) can be enforced by ensuring that if two distinct vertices
  $v,v' \in V(G)$ are connected in $G[Y]$ then $vv' \in Y$. \nb{Iyad: I think it would be better to use $vv'$ to denote the edge. I changed it as such.}
  \nb{Agreed - we should do that. -R+S}
  This can be enforced by the formula $\textfor{FP1}$ defined as:
  \begin{align*}
    \textfor{FP1}(Y) := & \forall v \forall v' (V v \land V v' \land v\neq v' \land
                      \textfor{conn}(v,v',Y)) \rightarrow \\
                       & (\exists y I yv \land I yv')
  \end{align*}
  Here the formula $\textfor{conn}(Y,v,v')$ is defined as follows:
  \begin{align*}
    \textfor{conn}(v,v',Y) := & \forall W (W v \land \lnot (W v')) \rightarrow \\
                              & (\exists x \exists y \exists e W x
                                \land V y \land \lnot W y \land Y e  \\
                              & \land I ex \land I ey)
  \end{align*}

  Provided that (P1) holds for $Y$, (P2) can now be enforced by ensuring that
  for every vertex $v \in V(G)$ either $v \in X$ or there is a edge
  $vu$ in $Y$ such that $u \in X$.
  The following formula, denoted by $\textfor{FP2}(Y,X)$, ensures (P2)
  using this idea:
  \begin{align*}
    \textfor{FP2}(Y,X) := & \forall v (V v \rightarrow X v \lor
                            (\exists u \exists e I eu \land Iev \land
                            X u))
  \end{align*}
  Putting everything together we obtain the
  formula $\Phi(X)$ as:
  \begin{align*}
    \Phi(X) := & \exists Y \textfor{FP1}(Y) \land \textfor{FP2}(Y,X)
  \end{align*}
  }
\end{proof}
}
Note that the above theorem also implies that the well-known
\textsc{Coloring} problem is \FPT{} parameterized
by the treewidth of the complement of the input graph.
The theorem below follows immediately from Lemmas~\ref{lem:drm-cc}
and~\ref{lem:cc-tw}.
\begin{THE}\label{the:drm-tw}
  \uDRM{} and \bDRM{} are in \FPT{} when parameterized by the treewidth of
  the compatibility graph.
\end{THE}

\subsection{\bDRM}

\begin{THE}\label{thm:bounded-drm-fpt}
  \bDRM\emph{[\comb{}]} is in \FPT{}.
\end{THE}
\begin{proof}
  Let $(\matM, t)$ be an instance of \bDRM{}, and let $k$ be the
  parameter $\comb$.
  % W.l.o.g. we can assume that the set $R$ of rows
  % of $\matM$ is distinct; since if $\matM$ contains two
  % identical rows then deleting one of them does not change the solution of
  % the instance.
 By Proposition~\ref{pro:compcomb}, we can
  compute a set $R_{\bullet}$ of rows and a set $C_{\bullet}$ of columns, where
  $|R_{\bullet}\cup C_{\bullet}|\leq k$, and such that every occurrence of $\bullet$ in
  $\matM$ is either contained in a row or column in $R_{\bullet}\cup
  C_{\bullet}$. Let $R$ and $C$ be the set of rows
  and columns of $\matM$, respectively. Let $\PPP$ be the unique partition
  of $R \setminus R_{\bullet}$ such that two rows $r$ and $r'$ belong to the
  same set in $\PPP$ if and only if they are identical on all columns
  in $C \setminus C_{\bullet}$. Then $|P|\leq (p+1)^k$, for every $P \in
  \PPP$, since two rows in $P$ can differ on at most $|C_{\bullet}|\leq k$
  entries, each having $(p+1)$ values to be chosen from. Moreover, any two
  rows in $R \setminus R_{\bullet}$ that are not contained in the same set
  in $\PPP$ are not compatible, which implies that they appear in different components of $\comG(\matM)\setminus R_{\bullet}$
  and hence the set of vertices in
  every component of $\comG(\matM)\setminus R_{\bullet}$ is a subset of $P$, for some $P \in
  \PPP$. It is now straightforward to show that $\tw(\comG(\matM))\leq
  k+(p+1)^k$, and hence,  $\tw(\comG(\matM))$ is bounded by a function of the parameter $k$.
  \lv{Towards showing this consider the tree-decomposition
    $(T,\chi)$ for $\comG(\matM)$, where $T$ is a path containing one
    node $t_P$ with $\chi(t_P)=R_{\bullet}\cup P$ for every $P \in \PPP$. Then
    $(T,\chi)$ is tree-decomposition of width $k+(p+1)^k-1$ for $\comG(\matM)$.}
  The theorem now follows from Theorem~\ref{the:drm-tw}.
\end{proof}

\subsection{\uDRM}

The proof of the following theorem is very similar to the proof of
Theorem~\ref{thm:bounded-drm-fpt}, i.e., we mainly use the observation
that the parameter \row\ is also a bound on the treewidth
of the compatibility graph and then apply Theorem~\ref{the:drm-tw}.
\begin{THE}\label{thm:unbounded-drm-fpt}
  \uDRM\emph{[\row{}]} is in \FPT{}.
\end{THE}
\lv{\begin{proof}
  Let $(p,\matM, t)$ be an instance of \uDRM{}, let $k$ be the
  parameter $\row$ and let $R_{?}$ be a set of rows with $|R_{?}|\leq
  k$ covering all occurrences of $\bullet$ in $\matM$.
  %  W.l.o.g. we can assume that the set $R$ of rows
  % of $\matM$ is distinct; since if $\matM$ contains two
  % identical rows then deleting one of them does not change the solution of
  % the instance.
  Then $\comG(\matM)\setminus R_{?}$ is an independent
  set since any two distinct rows without $\bullet$ are not compatible. It
  is now straightforward to show that $\tw(\comG(\matM))\leq k$ and
  hence bounded by a function of our parameter $k$.
  Towards showing this the following tree-decomposition
  $(T,\chi)$ for $\comG(\matM)$, where $T$ is a path containing one
  node $t_r$ with $\chi(t_r)=R_{?}\cup \{r\}$ for every $r \in R
  \setminus R_{?}$. Then
  $(T,\chi)$ is tree-decomposition of width $k$ for $\comG(\matM)$.
  The theorem now follows from Theorem~\ref{the:drm-tw}.
\end{proof}
}

\sv{For the remainder of this section, we will introduce the
\textsc{Partitioning Into Triangles} (\PIT{}) problem: Given a
graph $G$, decide whether there is a partition $\PPP$ of $V(G)$ into triangles.}
\lv{
For our remaining hardness proofs we will make use of the following problem.
\pbDef{\textsc{Partitioning Into Triangles} (\PIT{})}{A graph $G$.}{Is
  there a partition $\PPP$ of $V(G)$ into triangles, i.e., $G[P]$ is a
  triangle for every $P \in \PPP$?}
}
We will often use the following easy observation.
\begin{OBS}\label{obs:pit-pic}
  A graph $G$ that does not contain a clique with four vertices
  has a partition into triangles if and only if it has a
  partition into at most $|V(G)|/3$ cliques.
\end{OBS}

\begin{THE}\label{thm:unbounded-drm-np}
  \uDRM\emph{[\col{}]} is \paraNP\hy hard.
\end{THE}
\lv{\begin{proof}}
\sv{\begin{proof}[Proof Sketch.]}
  \sv{We will reduce from the \NP-complete \TSATT\ problem~\cite{BermanKarpinskiScott03j}: Given a
    propositional formula $\phi$ in conjunctive normal form such that
    (1) every clause of $\phi$ has exactly three distinct literals and
    (2) every literal occurs in exactly two clauses, decide whether
    $\phi$ is satisfiable.}
  \lv{
  We will reduce from the following variant of $3$-SAT, which is
  \NP\hy complete~\cite{BermanKarpinskiScott03j}.

  \pbDef{\textsc{$3$-Satisfiability-$2$} (\TSATT{})}{A propositional
    formula $\phi$ in conjunctive normal form such that (1) every clause
    of $\phi$ has exactly three distinct literals and (2) every literal occurs in exactly two clauses.}{Is $\phi$ satisfiable?}
  }
  To make our reduction easier to follow, we will divide the reduction
  into two steps. Given an instance (formula) $\phi$ of \TSATT{}, we will first
  construct an equivalent instance $G$ of \PIT{} with the additional
  property that $G$ does not contain a clique on four vertices. We note
  that similar reductions from variants of the satisfiability
  problem to \PIT{} are known (and hence our first step does not show
  anything new for \PIT{}); however, our reduction is specifically designed to
  simplify the second step, in which we will
  construct an instance $(\matM,|V(G)|/3)$ of \uDRM{} such that
  $\comG(\matM)$ is isomorphic to $G$ and $\matM$ has only seven
  columns. By Observation~\ref{obs:pit-pic} and Lemma~\ref{lem:drm-cc}, this proves the theorem
  since $(\matM,|V(G)|/3)$ has a solution if and only if $\phi$ does.

  Let $\phi$ be an instance of \TSATT{} with variables
  $x_1,\dotsc,x_n$ and clauses $C_1,\dotsc,C_m$. We first
  construct the instance $G$ of \PIT{} such that $G$ does not contain
  a clique of size four. \sv{For every variable $x_i$ of $\phi$, let $G(x_i)$ be the graph
    illustrated in Figure~\ref{fig:udrm-hard-gxi}, and for every clause
    $C_j$ of $\phi$, let $G(C_j)$ be the graph illustrated in
    Figure~\ref{fig:udrm-hard-gcj}.
    }\lv{For every variable $x_i$ of $\phi$, let $G(x_i)$ be the
  graph with vertices $x_i^1$, $x_i^2$, $\bar{x}_i^1$, $\bar{x}_i^2$,
  $x_i$ and edges forming a triangle on the vertices $x_i^1$, $x_i^2$, and $x_i$ as
  well as a triangle on the vertices $\bar{x}_i^1$, $\bar{x}_i^2$, and
  $x_i$. Moreover, for every clause $C_j$ with literals $l_{j,1}$,
  $l_{j,2}$, and $l_{j,3}$, let $G(C_j)$ be the graph with vertices
  $l_{j,1}^1$, $l_{j,1}^2$, $l_{j,2}^1$, $l_{j,2}^3$, $l_{j,3}^1$,
  $l_{j,3}^2$, $h_j^1$, and $h_j^2$ and edges between $l_{j,r}^1$ and
  $l_{j,r}^2$ for every $r \in \{1,2,3\}$ as well as edges forming a
  complete bipartite graph between $\{h_j^1,h_j^2\}$ and all other
  vertices of $G(C_j)$.
  }Let $f : [m]\times [3]\rightarrow \SB
  x_i^o,\bar{x}_i^o \SM 1 \leq i \leq n \land 1 \leq o \leq 2\SE$ be
  any bijective function such that for every $j$ and $r$ with $1 \leq
  j \leq m$ and $1\leq r \leq 3$, it holds that: If $f(j,r)=x_i^o$ (for
  some $i$ and $o$), then $x_i$ is the $r$-th literal of $C_j$; and
  if $f(j,r)=\bar{x}_i^o$, then $\bar{x}_i$ is the $r$-th
  literal of $C_j$. \lv{Figures~\ref{fig:udrm-hard-gxi}
  and~\ref{fig:udrm-hard-gcj} illustrate the gadgets $G(x_i)$ and $G(C_j)$.}

  The graph $G$ is obtained from the
  disjoint union of the graphs $G(x_1),\dotsc,G(x_n),G(C_1),\dotsc,G(C_m)$
  after applying the following modifications:
  \lv{\begin{itemize}
  \item}\sv{(1)} For every $j$ and $r$ with $1 \leq j \leq m$ and $1\leq r \leq
    3$ add edges forming a triangle on the vertices $l_{j,r}^1$,
    $l_{j,r}^2$, $f(j,r)$\lv{.
  \item}\sv{; and (2)} for every $i$ with $1\leq i \leq 2n-m$, add the vertices
    $g_i^1, g_i^2$ and an edge between $g_i^1$ and
    $g_i^2$. Finally we add edges forming a complete bipartite graph
    between all vertices in $\SB g_i^o \SM 1 \leq i \leq 2n-m \land 1 \leq o \leq 2 \SE$ and all vertices in $\SB
    h_i^o \SM 1 \leq i \leq n \land 1 \leq o \leq 2\SE$.
  \lv{\end{itemize}}

  \begin{figure}[tb]
    \centering
    % \resizebox{!}{\textwidth/2}{
    \begin{tikzpicture}[node distance=0.6cm]
      \tikzstyle{every node}=[]
      \tikzstyle{gn}=[circle, inner sep=2pt,draw]
      \tikzstyle{dots}=[circle, inner sep=1pt,draw]
      \tikzstyle{every edge}=[draw, line width=1.5pt]

      \draw
      node[gn, label=left:{$x_i(i,\bullet,0,0,0,0,0)$}] (xi) {}
      node[node distance=2cm, right of=xi] (ci) {}

      node[gn, node distance=0.5cm, above of=ci, label=right:{$x_i^2(i,1,\bullet,0,\bullet,0,0)$}] (xi2) {}
      node[gn, above of=xi2, label=right:{$x_i^1(i,1,\bullet,\bullet,0,0,0)$}] (xi1) {}
      node[gn, node distance=0.5cm, below of=ci, label=right:{$\bar{x}_i^2(i,0,\bullet,0,\bullet,0,0)$}] (bxi2) {}
      node[gn, below of=bxi2, label=right:{$\bar{x}_i^1(i,0,\bullet,\bullet,0,0,0)$}] (bxi1) {}

      (xi) edge (xi1)
      (xi) edge (xi2)
      (xi) edge (bxi1)
      (xi) edge (bxi2)
      (xi1) edge (xi2)
      (bxi1) edge (bxi2)
      ;
    \end{tikzpicture}
    % }
    \vspace{-0.1cm}\caption{An illustration of the gadget $G(x_i)$ introduced in the reduction of Theorem~\ref{thm:unbounded-drm-np}. The label of each vertex $v$ indicates the row vector $R(v)$.\vspace{-0.1cm}}
    \label{fig:udrm-hard-gxi}
  \end{figure}
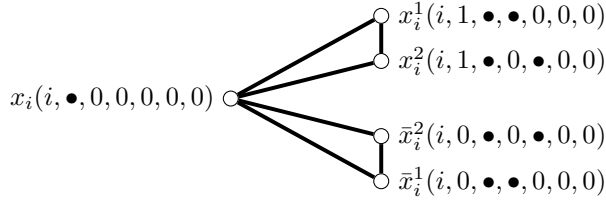

    \begin{figure}[tb]
    \centering
    % \resizebox{!}{\textwidth/2}{
    \begin{tikzpicture}[node distance=0.6cm]
      \tikzstyle{every node}=[]
      \tikzstyle{gn}=[circle, inner sep=2pt,draw]
      \tikzstyle{dots}=[circle, inner sep=1pt,draw]
      \tikzstyle{every edge}=[draw, line width=1.5pt]

      \draw
      node[gn, label=left:{$h_j^1(\bullet,\bullet,j,1,1,1,\bullet)$}] (hj1) {}
      node[gn, below of=hj1, label=left:{$h_j^2(\bullet,\bullet,j,1,1,2,\bullet)$}] (hj2) {}

      node[gn, node distance=2cm, right of=hj1, label=right:{$l_{j,2}^1(5,1,j,\bullet,1,\bullet,0)$}] (lj21) {}
      node[gn, node distance=2cm, right of=hj2, label=right:{$l_{j,2}^2(5,1,j,\bullet,1,\bullet,0)$}] (lj22) {}

      node[gn, above of=lj21, label=right:{$l_{j,1}^2(4,1,j,1,\bullet,\bullet,0)$}] (lj12) {}
      node[gn, above of=lj12, label=right:{$l_{j,1}^1(4,1,j,1,\bullet,\bullet,0)$}] (lj11) {}

      node[gn, below of=lj22, label=right:{$l_{j,3}^1(6,0,j,1,\bullet,\bullet,0)$}] (lj31) {}
      node[gn, below of=lj31, label=right:{$l_{j,3}^2(6,0,j,1,\bullet,\bullet,0)$}] (lj32) {}

      (lj11) edge (lj12)
      (lj21) edge (lj22)
      (lj31) edge (lj32)

      (hj1) edge (lj11)
      (hj1) edge (lj12)
      (hj1) edge (lj21)
      (hj1) edge (lj22)
      (hj1) edge (lj31)
      (hj1) edge (lj32)

      (hj2) edge (lj11)
      (hj2) edge (lj12)
      (hj2) edge (lj21)
      (hj2) edge (lj22)
      (hj2) edge (lj31)
      (hj2) edge (lj32)

      ;
    \end{tikzpicture}
    % }
    \vspace{-0.1cm}\caption{An illustration of the gadget $G(C_j)$ introduced in the reduction of Theorem~\ref{thm:unbounded-drm-np}. The label of each vertex $v$ indicates the row vector $R(v)$; here we assume that $f(j,1)=x_4^1$, $f(j,2)=x_5^2$, and $f(j,3)=\bar{x}_6^1$.\vspace{-0.1cm}}
    \label{fig:udrm-hard-gcj}
  \end{figure}
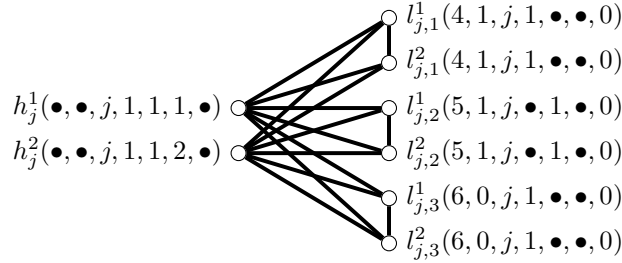
  This completes the construction of $G$. The following claim
  concludes the first step of our reduction.
  \begin{CLM}
    $\phi$ is satisfiable if and only if $G$ has a partition into
    triangles. Moreover, $G$ does not contain a clique of size four.
  \end{CLM}
\lv{\begin{proof}
      We first show that $G$ does
      not contain a clique of size four by showing that the neighborhood of any
      vertex in $G$ does not contain a triangle.
      \begin{itemize}
      \item If $v=x_i$ for some $i$ with $1 \leq i \leq n$, then
        $N_G(v)=\{x_i^1,x_i^2,\bar{x}_i^1,\bar{x}_i^2\}$ and does not
        contain a triangle.
      \item If $v=x_i^o$ for some $i$ and $o$ with
        $1\leq i \leq n$ and $1\leq o \leq 2$, then $N_G(v)=\{x_i, x_i^{o-1\mod
          2+1},l_{j,r}^1,l_{j,r}^2\}$, where $f^{-1}(x_i^o)=(j,r)$, and
        does not contain a triangle.
      \item The case for $v=\bar{x}_i^o$ for $i$ and $o$ as above is
        analogous.
      \item If $v=l_{j,r}^o$ for some $j$, $r$, and $o$ with $1 \leq j
        \leq m$, $1 \leq r \leq 3$, and $1\leq o \leq 2$, then
        $N_G(v)=\{l_{j,r}^{o-1\mod 2+1},f(j,r),h_j^1,h_j^2\}$ and does not
        contain a triangle.
      \item If $v=h_j^o$ for some $j$ and $o$ with $1 \leq j \leq m$ and
        $1\leq o \leq 2$, then $N_G(v)=\SB l_{j,r}^o\SM 1\leq r \leq 3
        \land 1 \leq o \leq 2\SE \cup \SB g_{j'}^1,g_{j'}^2\SM 1 \leq j'
        \leq 2n-m\SE$ and does not contain a triangle.
      \item If $v=g_j^o$ for some $j$ and $o$ with $1 \leq j \leq 2n-m$ and
        $1\leq o \leq 2$, then $N_G(v)=\{g_j^{o-1\mod 2+1}\}\cup \SB h_{j'}^{o'}\SM 1\leq j' \leq m
        \land 1 \leq o' \leq 2\SE$ and does not contain a triangle.
      \end{itemize}

      We now show that $\phi$ is satisfiable
      if and only if $G$ has a partition into triangles.
      Towards showing the forward direction let $\tau$ be a satisfying
      assignment for $\phi$. Then a partition $\PPP$ of $G$ into triangles
      contains the following triangles:
      \begin{itemize}
      \item[(1)] for every $i$ with $1\leq i \leq n$ the triangle $x_i$,
        $x_i^1$, $x_i^2$ if $\tau(x_i)=0$ and the triangle $x_i$,
        $\bar{x}_i^1$, $\bar{x}_i^2$, otherwise,
      \item[(2)] for every $j$ with $1 \leq j \leq m$ and every $r$ with $1\leq
        r \leq 3$ such that the $r$-th literal of $C_j$ is satisfied by
        $\tau$, the triangle $l_{j,r}^1,l_{j,r}^2,f(j,r)$,
      \item[(3)] for every $j$ with $1 \leq j \leq m$ and every $r$ with $1\leq
        r \leq 3$ such that the $r$-th literal of $C_j$ is not satisfied by
        $\tau$, the triangle $l_{j,r}^1,l_{j,r}^2,h_j^o$, where $o \in
        \{1,2\}$ and the $r$-th literal of $C_j$ is the $o$-th literal in
        $C_j$ that is not satisfied by $\tau$; note that this is always
        possible because $C_j$ has at most two literals that are not
        satisfied by $\tau$,
      \item[(4)] Let $A$ be the subset of $\SB h_i^o\SM 1\leq i \leq n \land
        1\leq o \leq 2\SE$ containing all $h_i^o$ that are not yet part of
        a triangle, i.e., that are not part of a triangle added in
        (3). Then $|A|=2n-m$ and it is hence possible to add the following
        triangles, i.e., for every $v \in A$ a triangle containing $v$ and the two
        vertices $g_p^1$ and $g_p^2$ for some $p$ with $1 \leq p \leq
        2n-m$.
      \end{itemize}

      Towards showing the reverse direction let $\PPP$ be a partition of
      $V(G)$ into $|V(G)|/3$ triangles. Then $\PPP$ satisfies:
      \begin{itemize}
      \item[(A1)] For every $i$ with $1\leq i \leq n$, $\PPP$ either
        contains the triangle $\{x_i,x_i^1,x_i^2\}$ or the triangle
        $\{x_i,\bar{x}_i^1,\bar{x}_i^2\}$.
        % \item[(2)] For every $j$ with $1\leq j \leq m$, $\PPP$ contains a triangle
        %   $\{h_j^1,l_{j,r},l_{j,r}\}$ for some $r \in \{1,2,3\}$ as well as
        %   a triangle $\{h_j^1,l_{j,r'},l_{j,r'}\}$ for some $r' \in
        %   \{1,2,3\}\setminus \{r'\}$.
      \item[(A2)] For every $j$ with $1\leq j \leq m$, $\PPP$
        there is an $r$ with $1\leq r \leq 3$ such that $\PPP$ contains
        the triangle $\{l_{j,r}^1,l_{j,r}^2,f(j,r)\}$.
      \end{itemize}
      (A1) follows because these are the only two triangles in $G$
      containing $x_i$ for every $i$ with $1 \leq i \leq n$. Moreover,
      (A2) follows because for every $j$ and $r$ with $1\leq j \leq m$ and
      $1\leq r \leq 3$ there are only three triangles containing one of
      the vertices $l_{j,r}^1$ and $l_{j,r}^2$, i.e., the triangles
      $\{l_{j,r}^1,l_{j,r}^2,h_j^1\}$, $\{l_{j,r}^1,l_{j,r}^2,h_j^2\}$,
      and $\{l_{j,r}^1,l_{j,r}^2,f(j,r)\}$.
      (A2) now follows
      because $\PPP$ can contain at most two triangles containing one of
      $h_i^1$ and $h_i^2$.
      But then the assignment $\tau$ setting $\tau(x_i)=1$ for every $i$
      with $1\leq i \leq n$ if and
      only if $\PPP$ contains the triangle
      $\{x_i,\bar{x}_i^1,\bar{x}_i^2\}$ is a satisfying assignment for
      $\phi$, because of (A2).
    \end{proof}
  }
  We will now proceed to the second (and final) step of our reduction, i.e., we will
  construct an instance $(\matM,|V(G)|/3)$ of \uDRM{} such that:
  $\comG(\matM)$ is isomorphic to $G$ and $\matM$ has only seven columns.
\lv{  Because of Observation~\ref{obs:pit-pic} and
  Lemma~\ref{lem:drm-cc}, this shows that \uDRM{} is \NP\hy hard
  already for only $7$ columns and concludes the proof of the theorem.}

\sv{$\matM$ contains one row $R(u)$ for every $u \in V(G)$. The
  definition of $R(u)$ for every vertex $u$ that is part of some
  gadget $G(x_i)$ or $G(C_j)$ is illustrated in
  Figures~\ref{fig:udrm-hard-gxi} and~\ref{fig:udrm-hard-gcj}. Additionally, we set
  $R(g_j^o)=(\bullet,\bullet,\bullet,\bullet,\bullet,\bullet,j)$ for
  every $j$ and $o$ with $1 \leq j \leq 2n-m$ and $1 \leq o \leq 2$.}
\lv{$\matM$ will contain one row $R(u)$ for every $u \in V(G)$, which is
  defined as follows.
  \begin{itemize}
  \item if $u=x_i$ for some $i$ with $1 \leq i \leq n$, we set
    $R(u)=(i,\bullet,0,0,0,0,0)$,
  \item if $u=x_i^o$ for some $i$ and $o$ with $1 \leq i \leq n$ and
    $1 \leq o \leq 2$, we set either:
    \begin{itemize}
    \item $R(u)=(i,1,\bullet,\bullet,0,0,0)$, if $o=1$ or
    \item $R(u)=(i,1,\bullet,0,\bullet,0,0)$, otherwise
    \end{itemize}
  \item if $u=\bar{x}_i^o$ for some $i$ and $o$ with $1 \leq i \leq n$ and
    $1 \leq o \leq 2$, we set either:
    \begin{itemize}
    \item $R(u)=(i,0,\bullet,\bullet,0,0,0)$, if $o=1$ or
    \item $R(u)=(i,0,\bullet,0,\bullet,0,0)$, otherwise
    \end{itemize}
  \item if $u=l_{j,r}^o$ for some $j$, $r$, and $o$ with $1 \leq j
    \leq n$, $1 \leq r \leq 3$, and $1 \leq o \leq 2$, then either:
    \begin{itemize}
    \item if $x_i$ is the $r$-th literal of $C_j$ (for some $i$ with $
      1\leq i \leq n$), then either:
      \begin{itemize}
      \item if $f(j,r)=x_i^1$, we set $R(u)=(i,1,j,1,\bullet,\bullet,0)$,
      \item otherwise, i.e., if $f(j,r)=x_i^2$, we set
        $R(u)=(i,1,j,\bullet,1,\bullet,0)$,
      \end{itemize}
    \item if $\bar{x}_i$ is the $r$-th literal of $C_j$ (for some $i$ with $
      1\leq i \leq n$), then either:
      \begin{itemize}
      \item if $f(j,r)=\bar{x}_i^1$, we set $R(u)=(i,0,j,1,\bullet,\bullet,0)$,
      \item otherwise, i.e., if $f(j,r)=\bar{x}_i^2$, we set
        $R(u)=(i,0,j,\bullet,1,\bullet,0)$,
      \end{itemize}
    \end{itemize}
  \item if $u=h_j^o$ for some $j$ and $o$ with $1 \leq j \leq m$ and
    $1 \leq o \leq 2$, we set
    $R(u)=(\bullet,\bullet,j,1,1,o,\bullet)$,
  \item if $u=g_j^o$ for some $j$ and $o$ with $1 \leq j \leq 2n-m$ and
    $1 \leq o \leq 2$, we set $R(u)=(\bullet,\bullet,\bullet,\bullet,\bullet,\bullet,j)$,
  \end{itemize}
  Figures~\ref{fig:udrm-hard-gxi}
  and~\ref{fig:udrm-hard-gcj} illustrate the row vectors assigned to
  the vertices in the gadgets $G(x_i)$ and $G(C_j)$.}\sv{Using an exhaustive case analysis, one can show that
    $\comG(\matM)$ is indeed isomorphic to $G$, which concludes the
    proof of the theorem.\qedhere}
  \lv{It remains to show that $\comG(\matM)$ is indeed isomorphic to $G$.
  Let $u \in V(G)$, we distinguish the following cases:
  \begin{itemize}
  \item if $u=x_i$ for some $i$ with $1 \leq i \leq n$, we need to show that
    $N_{\comG(\matM)}(u)=N_G(u)=\{x_i^1,x_i^2,\bar{x}_i^1,\bar{x}_i^2\}$.
    Since $R(x_i)=(i,\bullet,0,0,0,0,0)$ is compatible with
    $R(x_i^1)=(i,1,\bullet,\bullet,0,0,0)$, $R(x_i^2)=(i,1,\bullet,0,\bullet,0,0)$,
    $R(\bar{x}_i^1)=(i,0,\bullet,\bullet,0,0,0)$, and $R(\bar{x}_i^2)=(i,0,\bullet,0,\bullet,0,0)$, we
    already have that $N_G(u) \subseteq N_{\comG(\matM)}(u)$. Moreover,
    $R(x_i)=(i,\bullet,0,0,0,0,0)$ is not compatible with:
    \begin{itemize}
    \item $R(x_{i'})$, $R(x_{i'}^o)$, or $R(\bar{x}_{i'}^o)$ for any $i'$ and $o$ with
      $1\leq i'\leq n$, $i'\neq i$, and $1\leq o \leq 2$ because the
      first column of these rows is equal to $i'$ and $i'\neq i$.
    \item $R(l_{j,r}^o)$ for any $j$, $r$, and $o$ with $1 \leq j \leq
      m$, $1 \leq r \leq 3$, and $1\leq o \leq 2$ because the third
      column of $R(l_{j,r}^o)$ is equal to $j$, and
      hence not equal to the corresponding column of $R(u)$, which is
      $0$.
    \item $R(h_j^o)$ for any $j$ and $o$ with $1\leq j \leq m$ and $1
      \leq o \leq 2$, because the third column of
      $R(h_j^o)$ is equal to $j$ and $j \neq 0$.
    \item $R(g_j^o)$ for any $j$ and $o$ with $1\leq j \leq 2n-m$ and $1
      \leq o \leq 2$, because the seventh column of $R(g_j^o)$ is
      equal to $j$ and $j \neq 0$.
    \end{itemize}
    This shows that $N_{\comG(\matM)}(u) \subseteq N_G(u)$ and hence
    $N_{\comG(\matM)}(u)=N_G(u)$, as required.
  \item if $u=x_i^1$ for some $i$ with $1 \leq i \leq n$, we need to show that
    $N_{\comG(\matM)}(u)=N_G(u)=\{x_i, x_i^2,l_{j,r}^1,l_{j,r}^2\}$,
    where $j$ and $r$ are such that $f^{-1}(x_i^1)=(j,r)$.
    Since $R(x_i^1)=(i,1,\bullet,\bullet,0,0,0)$ is compatible with
    $R(x_i)=(i,\bullet,0,0,0,0,0)$, $R(x_i^2)=(i,1,\bullet,0,\bullet,0,0)$,
    $R(l_{j,r}^1)=R(l_{j,r}^2)=(i,1,j,1,\bullet,\bullet,0)$, we
    already have that $N_G(u) \subseteq N_{\comG(\matM)}(u)$. Moreover,
    $R(x_i^1)=(i,1,\bullet,\bullet,0,0,0)$ is not compatible with:
    \begin{itemize}
    \item $R(x_{i'})$, $R(x_{i'}^o)$, or $R(\bar{x}_{i'}^o)$ for any $i'$ and $o$ with
      $1\leq i'\leq n$, $i'\neq i$, and $1\leq o \leq 2$ because the
      first column of these rows is equal to $i'$ and $i'\neq i$.
    \item $R(\bar{x}_i^o)$ for any $o$ with $1 \leq o \leq 2$ because
      the second column of $R(\bar{x}_i^o)$ is equal to $0$ and
      $0\neq 1$.
    \item $R(l_{j',r'}^o)$ for any $j'$, $r'$, and $o$ with $1 \leq j' \leq
      m$, $1 \leq r' \leq 3$, $1\leq o \leq 2$, and $(j',r')\neq
      (j,r)$ because either:
      \begin{itemize}
      \item the $r'$ literal of $C_j$ is not $x_i$, in
        which case either the first or second column of $R(u)$ and
        $R(l_{j',r'}^o)$ differ or
      \item the $r'$ literal of $C_j$ is $x_i$, in
        which case $f(j',r')=x_i^2$ and the fifth column of
        $R(l_{j',r'}^o)$ is equal to $1$, and hence not equal to the
        fifth column of $R(u)$, which is $0$.
      \end{itemize}
    \item $R(h_j^o)$ for any $j$ and $o$ with $1\leq j \leq m$ and $1
      \leq o \leq 2$, because the fifth column of
      $R(h_j^o)$ is equal to $1$.
    \item $R(g_j^o)$ for any $j$ and $o$ with $1\leq j \leq 2n-m$ and $1
      \leq o \leq 2$, because the seventh column of $R(g_j^o)$ is
      equal to $j$ and $j \neq 0$.
    \end{itemize}
    This shows that $N_{\comG(\matM)}(u) \subseteq N_G(u)$ and hence
    $N_{\comG(\matM)}(u)=N_G(u)$, as required.
  \item the cases for $u\in \{x_i^2,\bar{x}_i^1,\bar{x}_i^2\}$ are
    analogous to the previous case,
  \item if $u=l_{j,r}^1$ for some $j$ and $r$ with $1\leq j \leq
    m$ and $1\leq r \leq 3$, we have to show that
    $N_{\comG(\matM)}(u)=N_G(u)=\{l_{j,r}^2, f(j,r), h_j^1,h_j^2\}$.
    Assume w.l.o.g. that $f(j,r)=x_i^1$ for some $i$ with $1
    \leq i \leq n$. Then since $R(l_{j,r}^1)=(i,1,j,1,\bullet,\bullet,0)$ is compatible with
    $R(l_{j,r}^2)=R(u)$, $R(x_i^1)=(i,1,\bullet,\bullet,0,0,0)$,
    $R(h_j^1)=(\bullet,\bullet,j,1,1,1,\bullet)$, and $R(h_j^2)=(\bullet,\bullet,j,1,1,2,\bullet)$,
    we
    already have that $N_G(u) \subseteq N_{\comG(\matM)}(u)$. Moreover,
    $R(l_{j,r}^1)=(i,1,j,1,\bullet,\bullet,0)$ is not compatible with:
    \begin{itemize}
    \item $R(x_i)=(i,\bullet,0,0,0,0,0)$ and $R(x_i^2)=(i,1,\bullet,0,\bullet,0,0)$ because
      of the fourth column,
    \item $R(\bar{x}_i^o)$ for any $o$ with $1 \leq o \leq 2$ because of the second column,
    \item $R(x_{i'})$, $R(x_{i'}^o)$, or $R(\bar{x}_{i'}^o)$ for any $i'$ and $o$ with
      $1\leq i'\leq n$, $i'\neq i$, and $1\leq o \leq 2$ because the
      first column of these rows is equal to $i'$ and $i'\neq i$.
    \item $R(l_{j',r'}^o)$ for any $j'$, $r'$, and $o$ with $1 \leq j'
      \leq m$, $1\leq r' \leq 3$, $(j',r')\neq (j,r)$, and $1\leq o
      \leq 2$ because either:
      \begin{itemize}
      \item $j'\neq j$ and they hence differ in the third column, or
      \item $j'=j$ but $r'\neq r$ in which case they differ in the
        first or second column because the clause $C_j$ cannot
        contain the same literal ($x_i$) twice.
      \end{itemize}
    \item $R(h_{j'}^o)$ for any $j'$ and $o$ with $1\leq j' \leq m$,
      $j'\neq j$, and $1
      \leq o \leq 2$, because the third column of
      $R(h_{j'}^o)$ is equal to $j'$ and $j'\neq j$.
    \item $R(g_{j'}^o)$ for any $j'$ and $o$ with $1\leq j' \leq 2n-m$ and $1
      \leq o \leq 2$, because the seventh column of $R(g_{j'}^o)$ is
      equal to $j'$ and $j'\neq 0$.
    \end{itemize}
  \item if $u=h_j^o$ for some $j$ and $o$ with $1\leq j \leq
    m$ and $1\leq o \leq 2$, we have to show that
    $N_{\comG(\matM)}(u)=N_G(u)=\{l_{j,r}^1,l_{j,r}^2\SM 1\leq r \leq
    3\SB \cup \SM g_{j'}^1,g_{j'}^2\SM 1\leq j'\leq 2n-m\SE$.
    Then since $R(h_j^o)=(\bullet,\bullet,j,1,1,o,\bullet)$ is compatible with
    $R(l_{j,r}^{o'})$ for any $r$ and $o'$ with $1 \leq r \leq 3$ and
    $1\leq o' \leq 2$ -- this is because the third column of
    $R(l_{j,r}^o)$ is always equal to $j$, the fourth and fifth
    columns are always either $\bullet$ or $1$, and the sixth column is
    always equal to $\bullet$ -- and $R(h_j^o)$ is also
    clearly compatible with $R(g_{j'}^1)=R(g_{j'}^2)=(\bullet,\bullet,\bullet,\bullet,\bullet,\bullet,j')$
    (for every $j'$ with $1 \leq j'\leq 2m-n$),
    we
    already have that $N_G(u) \subseteq N_{\comG(\matM)}(u)$. Moreover,
    $R(h_j^o)=(\bullet,\bullet,j,1,1,o,\bullet)$ is not compatible with:
    \begin{itemize}
    \item $R(x_i)$, $R(x_i^{o'})$, and $R(\bar{x}_i^{o'})$ for any $i$ and
      $o'$ with $1 \leq i \leq n$ and $1 \leq o' \leq 2$ because either
      the fourth or the fifth column of all these rows is equal to $0$
      and $0\neq 1$.
    \item $R(l_{j',r'}^{o'})$ for any $j'$, $r'$, and $o'$ with $1 \leq
      j' \leq m$, $j'\neq j$, $1 \leq r' \leq 3$ and $1 \leq o' \leq 2$
      because the third column of all these rows is $j'$ and $j'\neq
      j$,
    \item $R(h_j^{o-1\mod 2+1})$ because of the sixth column,
    \item $R(h_{j'}^{o'})$ for any $j'$ and $o'$ with $1 \leq j' \leq
      m$ and $1 \leq o'\leq 2$ because the third column of all these
      vectors is $j'$ and $j'\neq j$,
    \end{itemize}
  \item if $u=g_j^o$ for some $j$ and $o$ with $1\leq j \leq
    2n-m$ and $1\leq o \leq 2$, we have to show that
    $N_{\comG(\matM)}(u)=N_G(u)=\{g_j^{o-1\mod 2+1}\}\cup \SB h_{j'}^1,h_{j'}^2\SM 1\leq j' \leq
    m\SE$.
    We have already shown in the previous case that $g_j^o$ is
    adjacent to all vertices in
    $\SB h_{j'}^1,h_{j'}^2\SM 1\leq j' \leq
    m\SE$. Moreover, since $R(g_j^o)=R(g_j^{o-1\mod
      2+1})=(\bullet,\bullet,\bullet,\bullet,\bullet,\bullet,j)$ we obtain that
    $N_G(u) \subseteq N_{\comG(\matM)}(u)$. Moreover,
    $R(g_j^o)=(\bullet,\bullet,\bullet,\bullet,\bullet,\bullet,j)$ is not compatible with any other row
    because the seventh column of any other row is not equal to $j$.\qedhere
  \end{itemize}
  }
\end{proof}

We conclude this section with a hardness result showing that
$2$-\textsc{DRMC} remains \NP-hard when the number of missing or known
entries in each column/row is bounded.
\begin{THE}
  \label{thm:drmchard}
  The restriction of $2$-\textsc{DRMC} to instances in which each row
  and each column contains exactly three missing entries is \NP\hy
  hard. The same holds for the restriction of $2$-\textsc{DRMC} to
  instances in which each row and each column contains at most 4 determined entries.
\end{THE}
\sv{
\begin{proof}[Proof Sketch.]

  Consider the problem of (properly) coloring a graph on $n$ vertices, having minimum degree $n-4$ and no independent set of size 4, by $n/3$ colors, where $n$ is divisible by 3;  denote this problem as \coloring{}.
  This problem is \NP-hard via a reduction from the {\sc Partition into Triangles} problem on $K_4$-free cubic graphs. The \NP-hardness of the latter problem follows from the \NP-hardness of the {\sc Partition into Triangles} problem on planar cubic graphs~\cite{reed}, since a $K_4$ in a cubic graph must be isolated, and hence can be removed from the start. Finally, using Observation~\ref{obs:pit-pic}, {\sc Partition into Triangles} on $K_4$-free cubic graphs is polynomial-time reducible to \coloring{}, via the simple reduction that complements the edges of the graph.

   Now we reduce from \coloring{} to \bDRM{} by mimicking a standard reduction from 3-coloring to rank minimization~\cite{peeters96}. Given an instance $G$ of \coloring{}, we construct an $n \times n$ matrix $\matM$  whose rows and columns correspond to the vertices in $G$, as follows. The diagonal entries of $\matM$ are all ones. For an entry at row $i$ and column $j$, where $i \neq j$, $\matM[i, j]=0$ if $ij \in E(G)$, and is $\bullet$ otherwise. Finally, we set $r=n/3$. Observe that since each vertex in $G$ has $n-4$ neighbors, the number of missing entries in any row and any column of $\matM$ is 3.

It is not difficult to show that $G$ is a yes-instance of \coloring{} if and only if $(\matM, n/3)$ is a yes-instance of $2$-\textsc{DRMC}.

The second statement in the theorem follows via a reduction from {\sc 3-Coloring} on graphs of maximum degree at most 4~\cite{gjs76}, using similar arguments.
\end{proof}
}

\lv{
\begin{proof}

  Consider the problem of (properly) coloring a graph on $n$ vertices, having minimum degree $n-4$ and no independent set of size 4, by $n/3$ colors, where $n$ is divisible by 3;  denote this problem as \coloring{}.

  We first show that this problem is \NP-hard via a reduction from
  the {\sc Partition into Triangles} problem on $K_4$-free cubic graphs. The \NP-hardness of the aforementioned problems follows from the \NP-hardness of the {\sc Partition into Triangles} problem on planar cubic graphs~\cite{reed}, since a $K_4$ in a cubic graph must be isolated, and hence can be removed from the start.

Next, observe that {\sc Partition into Triangles} on $K_4$-free cubic graphs is polynomial-time reducible to \coloring{}, via the simple reduction that complements the edges of the graph. Given a $K_4$-free cubic graph $G$, the (edge) complement of $G$, $\overline{G}$, has minimum degree $n-4$, and contains no independent set of size more than 3. This can be seen as follows. Clearly, if $G$ can be partitioned into triangles then $\overline{G}$ has an $(n/3)$-coloring. Conversely, if $\overline{G}$ has an $(n/3)$-coloring, then since $\overline{G}$ has no independent set of size more than 3, each of the $n/3$ color classes in $\overline{G}$ contains exactly 3 vertices. Therefore, $G$ can be partitioned into triangles.

Finally, we reduce from \coloring{} to \bDRM{} by mimicking a standard reduction from 3-coloring to rank minimization~\cite{peeters96}. Given an instance $G$ of \coloring{}, we construct an $n \times n$ matrix $\matM$  whose rows and columns correspond to the vertices in $G$, as follows. The diagonal entries of $\matM$ are all ones. For any entry at row $i$ and column $j$, where $i \neq j$, $\matM[i, j]=0$ if $ij \in E(G)$, and is $\bullet$ otherwise. Finally, we set $r=n/3$. Observe that since each vertex in $G$ has $n-4$ neighbors, the number of missing entries in any row of $\matM$ is 3. Since $\matM$ is symmetric, the number of missing entries in any column of $\matM$ is 3 as well.

If $G$ is a yes-instance of \coloring{}, then $G$ can be partitioned into $n/3$ independent sets, each of size exactly 3. We claim that the missing entries in $\matM$ can be filled so that the total number of identical rows is at most $n/3$. To shows this, it suffices to show that the three rows corresponding to the vertices in any of the $n/3$ independent sets in $G$ can be filled so that to produce identical rows. To do so, consider such an independent set $\{u, v, w\}$ in $G$, and for convenience, denote their corresponding rows in $\matM$ by $u, v, w$, respectively. Since $u, v, w$ is an independent set, the two entries $\matM[v,u]$ and $\matM[w, u]$ are $\bullet$. Similarly, we have $\matM[u,v] = \matM[w, v]=$$\bullet$ and $\matM[u,w] = \matM[v, w]=$ $\bullet$. (That is, the only entry containing ``1'' in any of the three rows has corresponding entries in the other two rows that are $\bullet$.) Therefore, if we replace all these six missing entries with 1, and all other missing entries with 0, we make the three rows $u, v, w$ identical.

To prove the converse, suppose that the missing entries in $\matM$ can be completed so to obtain at most $n/3$ identical rows. Note that, for any two adjacent vertices $u, v$, their rows cannot be completed into identical rows since $u$ has 0 at column $v$, whereas $v$ has 1. Therefore, all rows that are completed into the same identical row correspond to an independent set in $G$. Since $G$ has no independent set of size more than 3, and since the number of identical rows in the completed matrix is at most $n/3$, it follows that the number of rows that have been completed into the same row is exactly 3, and those correspond to an independent set of size 3 in $G$. Therefore, $G$ is a yes-instance of \coloring{}.

The second statement in the theorem follows via a reduction from {\sc 3-Coloring} on graphs of maximum degree at most 4, which is known to be \NP-hard~\cite{gjs76}, using similar arguments.
\end{proof}
}
\section{Conclusion}
\label{sec:conclusion}
We studied the parameterized complexity of two fundamental matrix
completion problems under several parameterizations. For the bounded
domain case, we painted a positive picture by showing that the two
problems are in \FPT{} (resp.~\RFPT{}) w.r.t.~all considered parameters. For the
unbounded domain case, we characterized the parameterized complexity
of \uDRM{} by showing that it is in \FPT{} parameterized by \row, and \paraNP-hard{} parameterized by \col{} (and hence by \comb). For \uRM, we could show its membership in \XP{} (resp.~\RXP{}) w.r.t.~all considered parameters. Three immediate open questions ensue:
\begin{itemize}
\item Is it possible to obtain a deterministic algorithm for \bRM{} and
  \uRM{} parameterized by \comb{}?
  % This would e.g. be the case if the
  % algorithms stated in Proposition~\ref{pro:quadratic} could be derandomized.
\item Can we improve our \XP{} (resp. \RXP{}) results for \uRM{} to
  \FPT{} (resp. \RFPT{}) or show that the problems are \W{1}\hy hard?
\item Does a hardness result, similar to the one given in
  Theorem~\ref{thm:drmchard} for \bDRM{}, hold for \bRM{}?
\end{itemize}

%\newpage

\bibliography{literature}

\begin{thebibliography}{31}
\providecommand{\natexlab}[1]{#1}
\providecommand{\url}[1]{\texttt{#1}}
\expandafter\ifx\csname urlstyle\endcsname\relax
  \providecommand{\doi}[1]{doi: #1}\else
  \providecommand{\doi}{doi: \begingroup \urlstyle{rm}\Url}\fi

\bibitem[net()]{netflix}
Matrix completion and the {N}etflix challenge.
\newblock See \url{https://en.wikipedia.org/wiki/Matrix_completion.}

\bibitem[Berman et~al.(2003)Berman, Karpinski, and
  Scott]{BermanKarpinskiScott03j}
Berman, Piotr, Karpinski, Marek, and Scott, Alex~D.
\newblock Approximation hardness of short symmetric instances of {MAX-3SAT}.
\newblock \emph{Electronic Colloquium on Computational Complexity {(ECCC)}},
  \penalty0 (049), 2003.

\bibitem[Bessiere et~al.(2008)Bessiere, Hebrard, Hnich, Kiziltan, Quimper, and
  Walsh]{BessiereHHKQW08}
Bessiere, Christian, Hebrard, Emmanuel, Hnich, Brahim, Kiziltan, Zeynep,
  Quimper, Claude{-}Guy, and Walsh, Toby.
\newblock The parameterized complexity of global constraints.
\newblock In Fox, Dieter and Gomes, Carla~P. (eds.), \emph{Proceedings of the
  Twenty-Third {AAAI} Conference on Artificial Intelligence, {AAAI} 2008,
  Chicago, Illinois, USA, July 13-17, 2008}, pp.\  235--240. {AAAI} Press,
  2008.

\bibitem[Bodlaender(1996)]{Bodlaender96}
Bodlaender, Hans~L.
\newblock A linear-time algorithm for finding tree-decompositions of small
  treewidth.
\newblock \emph{SIAM J. Comput.}, 25\penalty0 (6):\penalty0 1305--1317, 1996.

\bibitem[Bodlaender \& Koster(2008)Bodlaender and Koster]{BodlaenderKoster08}
Bodlaender, Hans~L. and Koster, Arie M. C.~A.
\newblock Combinatorial optimization on graphs of bounded treewidth.
\newblock \emph{The Computer Journal}, 51\penalty0 (3):\penalty0 255--269,
  2008.

\bibitem[Bodlaender et~al.(2016)Bodlaender, Drange, Dregi, Fomin, Lokshtanov,
  and Pilipczuk]{BodlaenderDDFLP16}
Bodlaender, Hans~L., Drange, P{\aa}l~Gr{\o}n{\aa}s, Dregi, Markus~S., Fomin,
  Fedor~V., Lokshtanov, Daniel, and Pilipczuk, Michal.
\newblock A ${O}(c^k n)$ 5-approximation algorithm for treewidth.
\newblock \emph{{SIAM} J. Comput.}, 45\penalty0 (2):\penalty0 317--378, 2016.

\bibitem[Bonnet et~al.(2017)Bonnet, Gaspers, Lambilliotte, R{\"{u}}mmele, and
  Saffidine]{BonnetGLRS17}
Bonnet, {\'{E}}douard, Gaspers, Serge, Lambilliotte, Antonin, R{\"{u}}mmele,
  Stefan, and Saffidine, Abdallah.
\newblock The parameterized complexity of positional games.
\newblock In Chatzigiannakis, Ioannis, Indyk, Piotr, Kuhn, Fabian, and
  Muscholl, Anca (eds.), \emph{44th International Colloquium on Automata,
  Languages, and Programming, {ICALP} 2017, July 10-14, 2017, Warsaw, Poland},
  volume~80 of \emph{LIPIcs}, pp.\  90:1--90:14. Schloss Dagstuhl -
  Leibniz-Zentrum fuer Informatik, 2017.

\bibitem[Cand{\`{e}}s \& Plan(2010)Cand{\`{e}}s and Plan]{cp10}
Cand{\`{e}}s, Emmanuel~J. and Plan, Yaniv.
\newblock Matrix completion with noise.
\newblock \emph{Proceedings of the {IEEE}}, 98\penalty0 (6):\penalty0 925--936,
  2010.

\bibitem[Cand{\`{e}}s \& Recht(2009)Cand{\`{e}}s and Recht]{cr09}
Cand{\`{e}}s, Emmanuel~J. and Recht, Benjamin.
\newblock Exact matrix completion via convex optimization.
\newblock \emph{Foundations of Computational Mathematics}, 9\penalty0
  (6):\penalty0 717--772, 2009.

\bibitem[Cand{\`{e}}s \& Tao(2010)Cand{\`{e}}s and Tao]{ct10}
Cand{\`{e}}s, Emmanuel~J. and Tao, Terence.
\newblock The power of convex relaxation: near-optimal matrix completion.
\newblock \emph{{IEEE} Trans. Information Theory}, 56\penalty0 (5):\penalty0
  2053--2080, 2010.

\bibitem[Cerioli et~al.(2008)Cerioli, Faria, Ferreira, Martinhon, Protti, and
  Reed]{reed}
Cerioli, M{\'{a}}rcia~R., Faria, Lu{\'{e}}rbio, Ferreira, Talita~O., Martinhon,
  Carlos A.~J., Protti, F{\'{a}}bio, and Reed, Bruce~A.
\newblock Partition into cliques for cubic graphs: Planar case, complexity and
  approximation.
\newblock \emph{Discrete Applied Mathematics}, 156\penalty0 (12):\penalty0
  2270--2278, 2008.

\bibitem[Courtois et~al.(2002)Courtois, Goubin, Meier, and
  Tacier]{CourtoisGMT02}
Courtois, Nicolas, Goubin, Louis, Meier, Willi, and Tacier, Jean{-}Daniel.
\newblock Solving underdefined systems of multivariate quadratic equations.
\newblock In \emph{Public Key Cryptography, 5th International Workshop on
  Practice and Theory in Public Key Cryptosystems, {PKC} 2002, Paris, France,
  February 12-14, 2002, Proceedings}, volume 2274 of \emph{Lecture Notes in
  Computer Science}, pp.\  211--227. Springer, 2002.

\bibitem[Cygan et~al.(2015)Cygan, Fomin, Kowalik, Lokshtanov, Marx, Pilipczuk,
  Pilipczuk, and Saurabh]{CyganFKLMPPS15}
Cygan, Marek, Fomin, Fedor~V., Kowalik, Lukasz, Lokshtanov, Daniel, Marx,
  D{\'{a}}niel, Pilipczuk, Marcin, Pilipczuk, Michal, and Saurabh, Saket.
\newblock \emph{Parameterized Algorithms}.
\newblock Springer, 2015.

\bibitem[Downey \& Fellows(2013)Downey and Fellows]{DowneyFellows13}
Downey, Rodney~G. and Fellows, Michael~R.
\newblock \emph{Fundamentals of Parameterized Complexity}.
\newblock Texts in Computer Science. Springer Verlag, 2013.
\newblock ISBN 978-1-4471-5558-4, 978-1-4471-5559-1.

\bibitem[Elhamifar \& Vidal(2013)Elhamifar and Vidal]{ev13}
Elhamifar, Ehsan and Vidal, Ren{\'{e}}.
\newblock Sparse subspace clustering: Algorithm, theory, and applications.
\newblock \emph{{IEEE} Trans. Pattern Anal. Mach. Intell.}, 35\penalty0
  (11):\penalty0 2765--2781, 2013.

\bibitem[Fazel(2002)]{fazel01}
Fazel, Maryam.
\newblock \emph{Matrix rank minimization with applications}.
\newblock PhD thesis, Stanford University, 2002.

\bibitem[Frieze et~al.(2004)Frieze, Kannan, and Vempala]{fkv04}
Frieze, Alan~M., Kannan, Ravi, and Vempala, Santosh.
\newblock Fast monte-carlo algorithms for finding low-rank approximations.
\newblock \emph{J. {ACM}}, 51\penalty0 (6):\penalty0 1025--1041, 2004.

\bibitem[Ganian \& Ordyniak(2018)Ganian and Ordyniak]{GanianOrdyniak18}
Ganian, Robert and Ordyniak, Sebastian.
\newblock The complexity landscape of decompositional parameters for {ILP}.
\newblock \emph{Artificial Intelligence}, 257:\penalty0 61 -- 71, 2018.

\bibitem[Garey et~al.(1976)Garey, Johnson, and Stockmeyer]{gjs76}
Garey, Michael~R., Johnson, David~S., and Stockmeyer, Larry~J.
\newblock Some simplified {NP}-complete graph problems.
\newblock \emph{Theoretical Computer Science}, 1\penalty0 (3):\penalty0
  237--267, 1976.

\bibitem[Gottlob et~al.(2010)Gottlob, Pichler, and Wei]{GottlobPichlerWei10}
Gottlob, Georg, Pichler, Reinhard, and Wei, Fang.
\newblock Bounded treewidth as a key to tractability of knowledge
  representation and reasoning.
\newblock \emph{Artif. Intell.}, 174\penalty0 (1):\penalty0 105--132, 2010.

\bibitem[Hardt et~al.(2014)Hardt, Meka, Raghavendra, and Weitz]{hmrw14}
Hardt, Moritz, Meka, Raghu, Raghavendra, Prasad, and Weitz, Benjamin.
\newblock Computational limits for matrix completion.
\newblock In \emph{Proceedings of The 27th Conference on Learning Theory},
  volume~35 of \emph{{JMLR} Workshop and Conference Proceedings}, pp.\
  703--725. JMLR.org, 2014.

\bibitem[Ibarra et~al.(1982)Ibarra, Moran, and Hui]{IbarraMH82}
Ibarra, Oscar~H., Moran, Shlomo, and Hui, Roger.
\newblock A generalization of the fast {LUP} matrix decomposition algorithm and
  applications.
\newblock \emph{J. Algorithms}, 3\penalty0 (1):\penalty0 45--56, 1982.

\bibitem[Keshavan et~al.(2010{\natexlab{a}})Keshavan, Montanari, and Oh]{kmo10}
Keshavan, Raghunandan~H., Montanari, Andrea, and Oh, Sewoong.
\newblock Matrix completion from a few entries.
\newblock \emph{{IEEE} Trans. Information Theory}, 56\penalty0 (6):\penalty0
  2980--2998, 2010{\natexlab{a}}.

\bibitem[Keshavan et~al.(2010{\natexlab{b}})Keshavan, Montanari, and
  Oh]{kmo101}
Keshavan, Raghunandan~H., Montanari, Andrea, and Oh, Sewoong.
\newblock Matrix completion from noisy entries.
\newblock \emph{Journal of Machine Learning Research}, 11:\penalty0 2057--2078,
  2010{\natexlab{b}}.

\bibitem[Kloks(1994)]{Kloks94}
Kloks, T.
\newblock \emph{Treewidth: Computations and Approximations}.
\newblock Springer Verlag, Berlin, 1994.

\bibitem[Miura et~al.(2014)Miura, Hashimoto, and Takagi]{MiuraHT14}
Miura, Hiroyuki, Hashimoto, Yasufumi, and Takagi, Tsuyoshi.
\newblock Extended algorithm for solving underdefined multivariate quadratic
  equations.
\newblock \emph{{IEICE} Transactions}, 97-A\penalty0 (6):\penalty0 1418--1425,
  2014.

\bibitem[Peeters(1996)]{peeters96}
Peeters, Ren{\'{e}}.
\newblock Orthogonal representations over finite fields and the chromatic
  number of graphs.
\newblock \emph{Combinatorica}, 16\penalty0 (3):\penalty0 417--431, 1996.

\bibitem[Recht(2011)]{recht11}
Recht, Benjamin.
\newblock A simpler approach to matrix completion.
\newblock \emph{Journal of Machine Learning Research}, 12:\penalty0 3413--3430,
  2011.

\bibitem[Robertson \& Seymour(1986)Robertson and Seymour]{RobertsonSeymour86}
Robertson, Neil and Seymour, P.~D.
\newblock Graph minors. {II}. {A}lgorithmic aspects of tree-width.
\newblock \emph{J. Algorithms}, 7\penalty0 (3):\penalty0 309--322, 1986.

\bibitem[Saunderson et~al.(2016)Saunderson, Fazel, and Hassibi]{sfh16}
Saunderson, James, Fazel, Maryam, and Hassibi, Babak.
\newblock Simple algorithms and guarantees for low rank matrix completion over
  {$F_2$}.
\newblock In \emph{Proceedings of The {IEEE} International Symposium on
  Information Theory}, pp.\  86--90. {IEEE}, 2016.

\bibitem[van Bevern et~al.(2016)van Bevern, Komusiewicz, Niedermeier, Sorge,
  and Walsh]{BevernKNSW16}
van Bevern, Ren{\'{e}}, Komusiewicz, Christian, Niedermeier, Rolf, Sorge,
  Manuel, and Walsh, Toby.
\newblock H-index manipulation by merging articles: Models, theory, and
  experiments.
\newblock \emph{Artif. Intell.}, 240:\penalty0 19--35, 2016.

\end{thebibliography}
\bibliographystyle{icml2018}
\end{document}